%% file: ms.tex
\documentclass[twoside,UKenglish,leqno,twocolumn]{article}
\usepackage[utf8]{inputenc}

\usepackage{ltexpprt}
\usepackage{enumerate}
\usepackage{amsmath}
\usepackage[pdfborderstyle={/S/U}]{hyperref}
\usepackage[capitalise]{cleveref}
\input{header}

\makeatletter
\crefname{@theorem}{Theorem}{Theorems}
\crefname{Definition}{Definition}{Definitions}
\crefname{fact}{Fact}{Facts}
\makeatother

\begin{document}

\title{Peeling Close to the Orientability Threshold –\\ Spatial Coupling in Hashing-Based Data Structures}
\author{Stefan Walzer\thanks{Technische Universität Ilmenau}}

\date{}

\maketitle






{
\small
\input{abstract}

}
           
\nocite{DR:Towards:2012}
\input{introduction}

\input{peeling-process}
\input{applying-richardson-urbanke}
\input{interpolating-wave-to-peeling}
\input{peeling-needs-orientability}
\input{experiments}
\input{conclusion}

\section{Acknowlegements}

{This work is a comprehensive overhaul of \cite{DW:DensePeelable:2019}, presented at ESA 2019. Back then, we were unaware of the general phenomenon of threshold saturation via spatial coupling \cite{KRU:Wave-Like:2015,KRU:Why-Conv-LDPC-Works:2010,KRU:Coupled-Achieves-Capacity:2013} (see also \cite{GMU:Maxwell-Conjecture-Via-Spatial-Coupling:2012,KTT:Phenomenological-Threshold-Improvement:2012,SB:Thresholds-via-Lyapunov:2013,YJNP:Simple-Proof-Threshold-Saturation:2012}). The old construction was similar, but the analysis was more ad-hoc. The results were weaker, less general and less elegant.}

{
    Shortly after \cite{DW:DensePeelable:2019} was published, Djamal Belazzougui recognised that we had stumbled upon a known phenomenon and provided very useful pointers.\\
    Konstantinos Panagiotou helped by locating additional relevant sources.\\
    Special thanks go to Martin Dietzfelbinger. He was a coauthor of \cite{DW:DensePeelable:2019} and closely followed the revision process. His valuable comments significantly improved the presentation of this work.
}

\bibliographystyle{plainurl}
\bibliography{bibliographie}

\end{document}

%% file: header.tex
\usepackage{microtype}
\usepackage{amsfonts}
\usepackage{mathtools,textcomp}
\usepackage{booktabs}

\input{unicode}

\usepackage{dsfont}
\newunicodechar{𝟙}{\mathds{1}}
\usepackage{pifont}
\newunicodechar{✗}{\ding{55}}
\usepackage{bm}
\usepackage{xspace}
\usepackage{silence}
\usepackage{arydshln}

\usepackage{tikz}
\usetikzlibrary{calc,arrows.meta,shapes,decorations.pathreplacing}

\makeatletter
\g@addto@macro\@floatboxreset\centering
\makeatother

\renewcommand{\O}{\mathcal{O}}

\DeclareMathOperator{\Po}{Po}
\DeclareMathOperator{\Bin}{Bin}

\let\P=\PP
\def\P{\mathbf{P}}
\def\hP{\mathbf{\hat{P}}}
\def\peel{\mathsf{peel}}
\def\root{\mathrm{root}}
\def\q#1{q^{(#1)}}

\def\co{\mathrm{co}}

\def\hut{\expandafter\hat}

\def\e{\mathrm{e}}
\def\U{\mathcal{U}}

\def\conv{\stackrel{r→∞}\longrightarrow}

\def\lipicsdescriptionlabel#1{{\setlength\labelsep{0pt}\descriptionlabel{#1}}}
\let\ldl=\lipicsdescriptionlabel
\def\itRef#1{\ldl{(\ref{#1})}}

%% file: unicode.tex
\usepackage{newunicodechar}
\usepackage{silence}
\newunicodechar{✓}{\checkmark}
\newunicodechar{⁻}{^-}
\newunicodechar{⁺}{^+}
\newunicodechar{₋}{_-}
\newunicodechar{₊}{_+}
\newunicodechar{ℓ}{\ell}
\newunicodechar{•}{\item}
\newunicodechar{…}{\dots}
\newunicodechar{≔}{\coloneqq}
\newunicodechar{≤}{\leq}
\newunicodechar{≥}{\geq}
\newunicodechar{≰}{\nleq}
\newunicodechar{≱}{\ngeq}
\newunicodechar{⊕}{\oplus}
\newunicodechar{⊗}{\otimes}
\newunicodechar{ϕ}{\varphi}
\newunicodechar{≠}{\neq}
\newunicodechar{¬}{\neg}
\newunicodechar{≡}{\equiv}
\newunicodechar{₀}{_0}
\newunicodechar{₁}{_1}
\newunicodechar{₂}{_2}
\newunicodechar{₃}{_3}
\newunicodechar{₄}{_4}
\newunicodechar{₅}{_5}
\newunicodechar{₆}{_6}
\newunicodechar{₇}{_7}
\newunicodechar{₈}{_8}
\newunicodechar{₉}{_9}
\newunicodechar{⁰}{^0}
\newunicodechar{¹}{^1}
\newunicodechar{²}{^2}
\newunicodechar{³}{^3}
\newunicodechar{⁴}{^4}
\newunicodechar{⁵}{^5}
\newunicodechar{⁶}{^6}
\newunicodechar{⁷}{^7}
\newunicodechar{⁸}{^8}
\newunicodechar{⁹}{^9}
\newunicodechar{∈}{\in}
\newunicodechar{∉}{\notin}
\newunicodechar{⊂}{\subset}
\newunicodechar{⊃}{\supset}
\newunicodechar{⊆}{\subseteq}
\newunicodechar{⊇}{\supseteq}
\newunicodechar{⊄}{\nsubset}
\newunicodechar{⊅}{\nsupset}
\newunicodechar{⊈}{\nsubseteq}
\newunicodechar{⊉}{\nsupseteq}
\newunicodechar{∪}{\cup}
\newunicodechar{∩}{\cap}
\newunicodechar{∀}{\forall}
\newunicodechar{∃}{\exists}
\newunicodechar{∄}{\nexists}
\newunicodechar{∨}{\vee}
\newunicodechar{∧}{\wedge}
\newunicodechar{ℝ}{\mathbb{R}}
\newunicodechar{ℕ}{\mathbb{N}}
\newunicodechar{𝔽}{\mathbb{F}}
\newunicodechar{ℤ}{\mathbb{Z}}
\newunicodechar{·}{\cdot}
\newunicodechar{∘}{\circ}
\newunicodechar{×}{\times}
\newunicodechar{↑}{\uparrow}
\newunicodechar{↓}{\downarrow}
\newunicodechar{→}{\rightarrow}
\newunicodechar{←}{\leftarrow}
\newunicodechar{⇒}{\Rightarrow}
\newunicodechar{⇐}{\Leftarrow}
\newunicodechar{↔}{\leftrightarrow}
\newunicodechar{⇔}{\Leftrightarrow}
\newunicodechar{↦}{\mapsto}
\newunicodechar{∅}{\emptyset}
\newunicodechar{∞}{\infty}
\newunicodechar{≅}{\cong}
\newunicodechar{≈}{\approx}

\newunicodechar{α}{\alpha}
\newunicodechar{β}{\beta}
\newunicodechar{γ}{\gamma}
\newunicodechar{Γ}{\Gamma}
\newunicodechar{δ}{\delta}
\newunicodechar{Δ}{\Delta}
\newunicodechar{ε}{\varepsilon}
\newunicodechar{ζ}{\zeta}
\newunicodechar{η}{\eta}
\newunicodechar{θ}{\theta}
\newunicodechar{Θ}{\Theta}
\newunicodechar{ι}{\iota}
\newunicodechar{κ}{\kappa}
\newunicodechar{λ}{\lambda}
\newunicodechar{Λ}{\Lambda}
\newunicodechar{μ}{\mu}
\newunicodechar{ν}{\nu}
\newunicodechar{ξ}{\xi}
\newunicodechar{Ξ}{\Xi}
\newunicodechar{π}{\pi}
\newunicodechar{Π}{\Pi}
\newunicodechar{ρ}{\rho}
\newunicodechar{σ}{\sigma}
\newunicodechar{Σ}{\Sigma}
\newunicodechar{τ}{\tau}
\newunicodechar{υ}{\upsilon}
\newunicodechar{ϒ}{\Upsilon}
\newunicodechar{ϕ}{\phi}
\newunicodechar{Φ}{\Phi}
\newunicodechar{χ}{\chi}
\newunicodechar{ψ}{\psi}
\newunicodechar{Ψ}{\Psi}
\newunicodechar{ω}{\omega}
\newunicodechar{Ω}{\Omega}
\newunicodechar{𝟙}{\mathds{1}}
\newunicodechar{✗}{\ding{55}}

%% file: abstract.tex

\begin{abstract}
    In multiple-choice data structures each element $x$ in a set $S$ of $m$ \emph{keys} is associated with a random set $e(x) ⊆ [n]$ of \emph{buckets} with capacity $ℓ ≥ 1$ by hash functions. This setting is captured by the hypergraph $H = ([n],\{e(x) \mid x ∈ S\})$. Accomodating each key in an associated bucket amounts to finding an \emph{$ℓ$-orientation} of $H$ assigning to each hyperedge an incident vertex such that each vertex is assigned at most $ℓ$ hyperedges. If each subhypergraph of $H$ has minimum degree at most $ℓ$, then an $ℓ$-orientation can be found greedily and $H$ is called \emph{$ℓ$-peelable}. Peelability has a central role in invertible Bloom lookup tables and can speed up the construction of retrieval data structures, perfect hash functions and cuckoo hash tables.

    Many hypergraphs exhibit sharp density thresholds with respect to $ℓ$-orientability and $ℓ$-peelability, i.e.\ as the density $c = \frac{m}{n}$ grows past a critical value, the probability of these properties drops from almost $1$ to almost $0$. In fully random $k$-uniform hypergraphs the thresholds $c_{k,ℓ}^*$ for $ℓ$-orientability significantly exceed the thresholds for $ℓ$-peelability. In this paper, for every $k ≥ 2$ and $ℓ ≥ 1$ with $(k,ℓ) ≠ (2,1)$ and every $z > 0$, we construct a new family of random $k$-uniform hypergraphs with i.i.d.\ random hyperedges such that \emph{both} the $ℓ$-peelability and the $ℓ$-orientability thresholds approach $c_{k,ℓ}^*$ as $z → ∞$. In particular we achieve $1$-peelability at densities arbitrarily close to $1$, extending the reach of greedy algorithms.
    
    Our construction is simple: The $n$ vertices are linearly ordered and each hyperedge selects its $k$ elements uniformly at random from a random range of $\frac{n}{z+1}$ consecutive vertices.
    
    We thus exploit the phenomenon of \emph{threshold saturation} via \emph{spatial coupling} discovered in the context of low-density parity-check codes. Once the connection to data structures is in plain sight, a framework by Kudekar, Richardson and Urbanke \cite{KRU:Wave-Like:2015} does the heavy lifting in our proof.
    
    We demonstrate the usefulness of our construction using our hypergraphs as a drop-in replacement in a retrieval data structure by Botelho et al.\ \cite{BPZ:Practical:2013}. This reduces memory usage from $≈1.23m$ bits to $≈1.12m$ bits (for input size $m$). Using $k > 3$ attains, at small sacrifices in running time, further improvements to memory usage.
\end{abstract}

%% file: introduction.tex

\def\cklo{c_{k,ℓ}^*}
\def\co#1#2{c_{#1,#2}^*}
\def\cklp{c_{k,ℓ}^{\scalebox{0.6}{$\triangle$}}}
\def\cpnoindex{c^{\scalebox{0.6}{$\triangle$}}}
\def\cp#1#2{c_{#1,#2}^{\scalebox{0.6}{$\triangle$}}}
\def\FRfamily{(H_{n,cn}^k)_{c ∈ ℝ₀^+,n ∈ ℕ}}
\def\fklzp{f_{k,ℓ,z}^{\scalebox{0.6}{$\triangle$}}}
\def\fklp{f_{k,ℓ}^{\scalebox{0.6}{$\triangle$}}}
\def\fklzo{f_{k,ℓ,z}^*}
\def\fklo{f_{k,ℓ}^*}
\def\hF{\hat{F}}
\def\convn{\stackrel{n→∞}{\longrightarrow}}
\def\refRel#1#2{\stackrel{\textrm{\lipicsdescriptionlabel{(\foreach \r[count=\i] in {#1}{\ifnum\i>1,\fi\ref{\r}})}}}{#2}}
\def\headMath#1#2{\texorpdfstring{$#1$}{#2}}
\def\construct{\textsf{construct}\xspace}
\def\eval{\textsf{eval}\xspace}
\def\member{\textsf{member}\xspace}

\section{Introduction}
\label{sec:intro}

Various data structures relying on the “power of multiple choices” \cite{M:The_Power:1991} feature a set of $n$ \emph{buckets}, indexed by $[n] = \{1,…,n\}$, each of capacity $ℓ ≥ 1$, and a set of $m$ \emph{keys}, each associated with $k ≥ 2$ random buckets via hash functions. Take \emph{cuckoo hashing} for instance. The task is to place each key into one of the $k$ buckets associated with it such that no bucket is overloaded. This can be modelled by a hypergraph $H = (V,E)$ with vertex set $V = [n]$ representing buckets and $m$ hyperedges of size $k$ representing keys. A placement of keys corresponds to an \emph{$ℓ$-orientation} of $H$ \cite{FKP:The_Multiple:2011}, which is a function $o \colon E → V$ with $o(e) ∈ e$ for all $e ∈ E$ and $|o^{-1}(v)| ≤ ℓ$ for all $v ∈ V$.

If each key $x$ is assigned its buckets $e(x) = \{h₁(x),h₂(x),…,h_k(x)\} ⊆ [n]$ by $k$ independent and fully random hash functions, then the corresponding \emph{fully random $k$-uniform hypergraph} is denoted by $H_{n,m}^k$ (the issue of repeated incidences $h_i(x) = h_j(x)$ or repeated hyperedges $e(x) = e(y)$ is irrelevant for our purposes).

For a family $(H_{n,c})_{c∈ℝ_{≥0},n∈ℕ}$ of random hypergraphs we say that $c^*$ is a threshold for a property $P$ if $c^* = \sup\{c ∈ ℝ_{≥0} \mid \lim_{n→∞}\Pr[H_{n,c} ∈ P] = 1\}$.
 The thresholds $c_{k,ℓ}^*$ for $ℓ$-orientability of $\FRfamily$ are known\footnote{
    They are also known to be \emph{sharp} (cf. \cite{Friedgut:Sharp-Thresholds:1999}), meaning $c^* = \inf\{ c ∈ ℝ_{≥0} \mid \Pr[H_{c,n} ∈ P] \convn 0 \}$ also holds.
}
  for all $k ≥ 2$ and $ℓ ≥ 1$, see \cite{PR:Cuckoo:2004,CSW:The_Random:2007,FR:The_k-orientability:2007,DGMMPR:Tight:2010,FM:Maximum:2012,FP:Sharp:2012,FKP:The_Multiple:2011}. They determine the limit $c_{k,ℓ}^*/ℓ$ of the memory efficiency (“used space $cn$ over allocated space $ℓn$”) that can be reliably achieved by cuckoo hash tables.

Some applications, however, rely on the stronger hypergraph property of \emph{$ℓ$-peelability}, which means that all subhypergraphs\footnote{A subhypergraph of $H = (V,E)$ is a hypergraph $H' =(V',E')$ with $V' ⊆ V$ and $E' ⊆ E ∩ 2^{V'}$.} of $H$ have minimum degree at most $ℓ$. This often enables greedy, linear time construction algorithms. To place all keys in a cuckoo hash table, repeatedly look for a bucket that is associated with at most $ℓ$ unplaced keys and place those keys in that bucket.

\begin{table*}[htb]
    \small
    \centering
    \renewcommand{\tabcolsep}{0.15cm}
    \begin{tabular}{cccccccc}
        \toprule 
        $ℓ$\textbackslash\raisebox{1.5pt}{$k$} & 2 & 3 & 4 & 5 & 6 & 7\\
        \midrule
        1 & (0.5\phantom{000}, 0\phantom{.0000}) & (0.9179, 0.8185) & (0.9768, 0.7723) & (0.9924, 0.7018) & (0.9974, 0.6371) & (0.9991, 0.5818)\\
        2 & (0.8970, 0.8377) & (0.9882, 0.7764) & (0.9982, 0.6668) & (0.9997, 0.5789) & (1.0000, 0.5108) & (1.0000, 0.4573)\\
        3 & (0.9592, 0.8582) & (0.9973, 0.7248) & (0.9998, 0.6036) & (1.0000, 0.5152) & (1.0000, 0.4496) & (1.0000, 0.3992)\\
        4 & (0.9804, 0.8499) & (0.9993, 0.6867) & (1.0000, 0.5624) & (1.0000, 0.4755) & (1.0000, 0.4123) & (1.0000, 0.3644)\\
        5 & (0.9896, 0.8365) & (0.9998, 0.6579) & (1.0000, 0.5331) & (1.0000, 0.4479) & (1.0000, 0.3867) & (1.0000, 0.3406)\\
        6 & (0.9941, 0.8229) & (0.9999, 0.6353) & (1.0000, 0.5108) & (1.0000, 0.4272) & (1.0000, 0.3677) & (1.0000, 0.3231)\\
        \bottomrule
    \end{tabular}
    \caption[fragile]{(Normalised) $ℓ$-orientability and $ℓ$-peelability thresholds $(\cklo/ℓ,\cklp/ℓ)$ of fully random $k$-uniform hypergraphs, rounded to four decimal places. The former quickly approach $1$ as $k$ or $ℓ$ increases, the latter do not.}
    \vspace{1em}
    \hrule
    \label{tab:threshold-comparison}
\end{table*}

The price to pay is typically reduced memory efficiency, since, at least when fully random hypergraphs are concerned, $ℓ$-peelability thresholds $\cklp$ fall short of $ℓ$-orientability thresholds $\cklo$ as shown in \cref{tab:threshold-comparison}. For a derivation of peelability thresholds, see \cite{Molloy05:Cores-in-random-hypergraphs,C:Cores:2004,PSW:SuddenCore:96,Luczak:A-simple-solution} and \cite[Chapter 18]{MezMont:InfPhysComp:2009}.

The main contribution of this paper is to propose a new kind of distribution for hyperedges, or equivalently, a new way to set up hash functions, that result in $k$-uniform hypergraphs with an $ℓ$-peelability threshold close to the $ℓ$-orientability threshold $c_{k,ℓ}^*$ of $\FRfamily$:
\begin{theorem}
    \label{thm:condensed}
    Let $k ≥ 2$ and $ℓ ≥ 1$ with $k+ℓ≥4$. For each $n$ and $c < c_{k,ℓ}^*$ there is a distribution $D_{n,c}^{k,ℓ}$ on $k$-subsets of $[n]$ such that the random hypergraph $\hF_{n,cn}^{k,ℓ}$ with vertex set $[n]$ and $m = \lfloor cn\rfloor$ hyperedges independently sampled according to $D_{n,c}^{k,ℓ}$ is $ℓ$-peelable with probability $1-o(1)$.
\end{theorem}

Before presenting our construction in \cref{sec:construction}, we explain the value of peelability in data structures and coding theory. In the latter field, the idea underlying our construction is known as “spatial coupling” and a suitable toolset for analysing corresponding thresholds already exists.

\Crefrange{sec:peeling-operator}{sec:upperbound} are devoted to proving our theorems, \cref{sec:experiments} presents experiments demonstrating the practical value of our approach.

\subsection{Data Structures Benefitting from Peelability}
\label{sec:HBDS}

\def\HBDS{\textsc{hbds}\xspace}
\def\LDPC{\textsc{ldpc}\xspace}

In the following, $S$ is always a set of $m = cn$ objects from some universe $\U$ and $[n]$ a set of buckets. Each $x ∈ S$ is associated with several buckets $e(x) ≔ \{h₁(x),…,h_k(x)\} ⊆ [n]$ via a constant\footnote{Some constructions allow $k = k(x)$ to depend on the key. \cite{LMSS:Efficient_Erasure:2001,R:Mixed:2013}} number $k ≥ 2$ of hash functions $h₁,…,h_k \colon \U → [n]$ (which we assume require a negligible amount of space to store). This gives rise to a hypergraph $H = ([n],\{e(x) \mid x ∈ S\})$. We review data structures benefitting from $ℓ$-peelability of $H$ (mostly for $ℓ = 1$).

\begin{description}
    •[Cuckoo Hash Table \cite{DW07:Balanced:2007,M:Some_Open:2009,PR:Cuckoo:2004}.] A cuckoo table implements a set or dictionary data structure with key set $S$. Each $x ∈ S$ (and, possibly, associated data) should be stored in exactly one bucket $o(x)$, and each bucket can hold up to $ℓ$ objects. To allow for constant-time lookups, we demand $o(x) ∈ e(x)$, which asks for an $ℓ$-orientation of $H$. If $H$ is $ℓ$-peelable, a greedy construction in linear time is possible.
    
    Otherwise, linear time constructions of $ℓ$-orientations are only known for the fully random hypergraph $H_{n,cn}^k$ with $c < \cklo$ in the following cases. For graphs (i.e.\ $k = 2$) and $ℓ ≥ 2$, linear time algorithms are described in \cite{CSW:The_Random:2007,FR:The_k-orientability:2007}. For $ℓ = 1$ and $k ≥ 3$, consider \cite{Khosla:Balls-Into-Bins:2013,Khosla:Balls-In-Bins-Extended:2019}. It is empirically plausible that random walk insertion can maintain an $ℓ$-orientation in a dynamic setting with expected constant time per update for any $k$ and $ℓ$ – a partial answer is given in \cite{FJ:Insertion-time-Cuckoo:2017}.
    •[Invertible Bloom Lookup Table (IBLT) \cite{GM:Invertable:2011}.] Among other things, IBLTs have been used to construct error correcting codes \cite{MV:Biff:2012} and solve the set reconciliation and straggler identification problem \cite{EG:StragglerIdentification:2011}. The data structure is inspired by Bloom Filters \cite{B:Space:1970} and Bloomier Filters \cite{CC:Bloomier_Filters:2008}.
    
    In IBLTs, each bucket $v ∈ [n]$ stores $\bigoplus_{x ∈ N(v)} x$, the bit-wise \textsc{xor} of (the bit representations of) the objects $N(v) ≔ \{ x ∈ S \mid v ∈ e(x)\}$ incident to $v$, as well as the degree $|N(v)|$. Note that this data structure is easy to maintain when insertions or deletions modify $S$, even through phases with $|S| \gg n$. Here, a \textsc{ListEntries} operation can be supported that recovers $S$ if $H$ is $1$-peelable and that fails otherwise.
    •[Retrieval \cite{DP:Succinct:2008,DW:Retrieval-log-extra-bits:2019,DW:One-Block-per-Row:2019,Vigna:Fast-Scalable-Construction-of-Functions:2016,P:An_Optimal:2009}.]
    An \emph{$r$-bit retrieval data structure} $D_f$ is constructed from a function $f\colon S → \{0,1\}^r$. The only operation $\eval$ must satisfy $\eval(D_f,x) = f(x)$ for all $x ∈ S$. The interesting setting is when $D_f$ may only occupy $\O(rm)$ bits. Note that naively storing $f$ as a set of pairs requires $m·(r+\log |\U|)$ bits. To save space, we exploit that the output of $\eval(D_f,y)$ for $y ∈ \U \setminus S$ may yield an arbitrary element of $\{0,1\}^r$ and that membership queries “$x ∈ S?$” need not be supported. 
    
    The idea is to find and store values $b₁,b₂,…,b_n ∈ \{0,1\}^r$ that satisfy the linear equations $\bigoplus_{i ∈ e(x)} b_i = f(x)$ for $x ∈ S$ (with $⊕$ denoting bit-wise \textsc{xor}). When found, the sequence $(b₁,…,b_n)$ is sufficient to answer \eval-queries and takes up $rn = rm/c$ bits as desired. 
    
    The existence of a solution $(b₁,…,b_n)$ is guaranteed if the incidence matrix of $H$ has rank $m$. Actually solving the linear system may take quadratic or cubic time. If $H$ is $1$-peelable, however, then the matrix is in row echelon form up to row and column exchanges and a solution can be found in linear time.\footnote{The matrix being in “triangular form” motivates our use of the symbol “$\cpnoindex$” for peeling thresholds.}
    
    Retrieval data structures are used in space efficient implementations of \emph{filters} and \emph{perfect hash functions} as follows.
    \begin{description}
        •[$\hookrightarrow$ XOR-Filters \cite{CC:Bloomier_Filters:2008,DP:Succinct:2008,ML:XorFilters:2019,Weaver:XORSAT-Filter}.]
        A filter for $S$ with false positive rate $ε > 0$ is a randomised data structure supporting a \member operation with $\member(x) = 1$ for $x ∈ S$ and $\Pr[\member(y) = 1] ≤ ε$ for $y ∉ S$. The best known example is the Bloom filter \cite{B:Space:1970,BM:Survey:2003,LGMRL:BloomSurvey:2019}. The following more space efficient alternative was appropriately dubbed “\textsc{xor}-filter”.
        
        We first select a random “fingerprint” function $f\colon \U → \{0,1\}^r$. Like all hash functions, we may assume that $f$ can be stored in $\O(1)$ space. However, we also store its restriction $f_S \colon S → \{0,1\}^r$ to $S$ explicitely in a retrieval data structure $D_{f_S}$ and define $\member(x) ≔ 𝟙[\eval(D_{f_S},x) = f(x)]$. For $x ∈ S$ we obtain $\member(x) = 1$ by construction. For $x ∉ S$, the independence of $D_{f_S}$ from $x$ and $f(x)$ guarantees that $\Pr[\member(x) = 1] = 2^{-r} = ε$.
        •[$\hookrightarrow$ Perfect Hashing \cite{BBOVV:Cache-Oblivious-Peeling:14,B:Near-Optimal,BPZ:Simple:2007,BPZ:Practical:2013,Vigna:Fast-Scalable-Construction-of-Functions:2016,MWHC:A_Family:1996}.] A \emph{perfect hash function} for a key set $S$ is an injective function $p \colon S → [n]$ with $n$ not much larger than $m = |S|$ such that $p$ is efficient to store and evaluate. A standard construction considers $H = H_{n,m}^4$ with $c = \frac mn < \co{4}{1} ≈ 0.977$ and a $1$-orientation $o\colon E → V$ of $H$ (which exists with high probability). Then $p(x) ≔ o(e(x))$ is injective. Since $o(e(x)) ∈ e(x) = \{h₁(x),…,h₄(x)\}$ there is $f \colon S → \{1,2,3,4\}$ such that $o(e(x)) = h_{f(x)}(x)$ for $x ∈ S$. Thus we only need to store $f$ with a (two-bit) retrieval data structure (as well as $h₁,…,h₄$) to be able to evaluate $p(x) = h_{f(x)}(x)$.
    \end{description}
\end{description}

\subsection{Connection to Coding Theory}

Our proof of \cref{thm:condensed} imports methods from coding theory. To explain the connection to this field, we briefly introduce the binary erasure channel and point out the relationship to our notions on hypergraphs. This reveals how closely the task of constructing good codes aligns with the task of constructing good hashing-based data structures.

\subsubsection{The Binary Erasure Channel and Low-Density~Parity-Check~Codes}
\label{back:BEC}

The \emph{binary erasure channel} (BEC) is a simple but important setting. We recommend \cite[Chapter 3]{RU:Modern-Coding-Theory:2008} for an excellent introduction to this subject.
When a sequence $(x₁,…,x_m) ∈ \{0,1\}^m$ is sent over the BEC, the receiver sees a sequence $(y₁,…,y_m) ∈ \{0,1,\textsf{?}\}^m$ where for each $i ∈ [m]$ independently, the $i$-th bit is \emph{erased} ($y_i = \textsf{?}$) with probability $ε ∈ [0,1]$ and unchanged ($y_i = x_i$) with probability $1-ε$. For reliable communication over such channels, redundancy is introduced. In \emph{linear codes}, several \emph{parity conditions} are each specified by a set $P ⊆ [m]$ and dictate that $\bigoplus_{i ∈ P} x_i$ is zero. The set of admissible messages (\emph{codewords}) then forms a linear subspace of $\{0,1\}^m$.

To relate this to hypergraphs, let $V$ be the set of all parity conditions and let $E^+ = \{e₁,…,e_m\}$ where $v ∈ e_i$ if $x_i$ is involved in parity condition $v$. The incidence graph of $H^+ = (V,E^+)$ is known as the \emph{Tanner graph} \cite{T:Tanner-Graph:81}. In \emph{low-density parity-check ({\upshape\LDPC}) codes} the Tanner graph is sparse.

During transmission, bits corresponding to some set $E ⊆ E^+$ are erased, and we consider $H = (V,E)$. When decoding, we seek an assignment $x_{\mathrm{dec}} \colon E → \{0,1\}$ such that for $v ∈ V$
we have $\bigoplus_{e ∈ E, e \ni v} x_{\mathrm{dec}}(e) = c_v$ where $c_v$ is the parity of the successfully transmitted bits involved in parity condition $v$. The existence of a solution is guaranteed by construction, namely $x_{\mathrm{dec}}(e_i) = x_i$ for $e_i ∈ E$. Uniqueness of the solution and thus success of the ideal \emph{maximum a posteriori probability decoder} (MAP-decoder) requires the kernel of the incidence matrix of $H$ to be trivial — a property that implies $1$-orientability.

Success of the linear time \emph{belief propagation decoder} (BP-decoder) requires $1$-peelability of $H$. This decoder repeatedly identifies a parity condition, where all but one of the involved bits are known, and then decodes the unknown bit.

\subsubsection{Good Codes vs.\ Good Data Structures}

The goals in hashing-based data structures (\HBDS) are sufficiently similar to those in \LDPC decoding to render the techniques from \LDPC codes useful in \HBDS.

\begin{description}
    •[Hyperedge Size.] The (average) hyperedge size $k$ is, in \HBDS, related to (average) query time and (average) number of cache faults per query. In \LDPC codes, $k$ is the (average) number of parity conditions relating to each message bit and contributes to overall encoding and decoding time. Thus, \emph{small $k$ is good}.
    •[Density.] In \HBDS, a high edge density $c = |E|/|V|$ means accommodating many objects in little space (high load), while in \LDPC codes it means recovering many erased bits from little redundancy (high rate). Thus, \emph{large $c$ is good}.
    •[Peelability.] As discussed, peelability is useful in both worlds, enabling greedy construction of data structures and greedy decoding of messages, respectively. 
    To the author's knowledge, $ℓ$-peelability for $ℓ > 1$ plays no role in coding theory, but the tools we import can handle this more general case nonetheless.
\end{description}

\noindent An aspect with imperfect alignment is \ldl{the role of randomness}.
Data structures must handle \emph{arbitrary} key sets and the role of hash functions is, in part, to make them behave like \emph{random} keys. Unless the hash function is itself a complex data structure, it is unreasonable to assume that this produces any \emph{useful} correlation. Therefore, the hypergraphs we deal with have independent random hyperedges (at best).

An \LDPC code on the other hand is given by a fixed hypergraph $H^+$. Randomness is frequently involved in its construction, but we are in principle free to design $H^+$, for instance, we might give all vertices the same degree. Since $H$ arises from a random $ε$-fraction of the hyperedges of $H^+$, this gives us control (proportional to $ε$) on $H$ as well. In this sense the hypergraphs relevant for \HBDS are a special case of those relevant for \LDPC codes.

\subsection{New Results}
\label{sec:construction}

In what we call “spatially coupled hypergraphs”, the vertices are linearly ordered and each hyperedge selects its $k$ elements uniformly at random from a random range of consecutive vertices.\footnote{In the precursor to this work \cite{DW:DensePeelable:2019} we partitioned the set of vertices into $z+k-1$ linearly ordered segments (for $z ∈ ℕ$). Each hyperedge selects a random range of $k$ consecutive segments and one random vertex from each of these segment. While this old construction has similar \emph{empirical} properties as the updated proposal, the discreteness in the construction leads to (slightly) inferior thresholds and does not seem to admit an elegant analysis.} 
Concretely:

\begin{Definition}[Spatially Coupled Hypergraph]
\label{def:hypergraphs}
    For $k ∈ ℕ$, $z,c ∈ ℝ₀^+$ and $n ∈ ℕ$, let $F_n~=~F(n,k,c,z)$ be the random $k$-uniform hypergraph with vertex set $V = [n]₀ = \{0,…,n-1\}$ and edge set $E$ of size $m = \lfloor cn\frac{z}{z+1}\rfloor$. Each edge $e ∈ E$ is independently obtained as
    $e = \big\{\big\lfloor\tfrac{y+o_i}{z+1}n\big\rfloor \mid i ∈ [k]\big\}$ where $y ∈ [\tfrac 12, z + \tfrac 12)$ and $o₁,…,o_k ∈ [-\tfrac 12,\tfrac 12]$ are chosen uniformly at random.
    
    We call $y ∈ Y ≔ [\tfrac 12, z + \tfrac 12)$ the position of $e ∈ E$ and $\frac{v(z+1)}{n} ∈ X = [0,z+1)$ the position of $v ∈ V$.
\end{Definition}

In \cref{fig:coupled-hypergraph} we sketch aspects of the construction.

\begin{figure}[hbt]
    \includegraphics{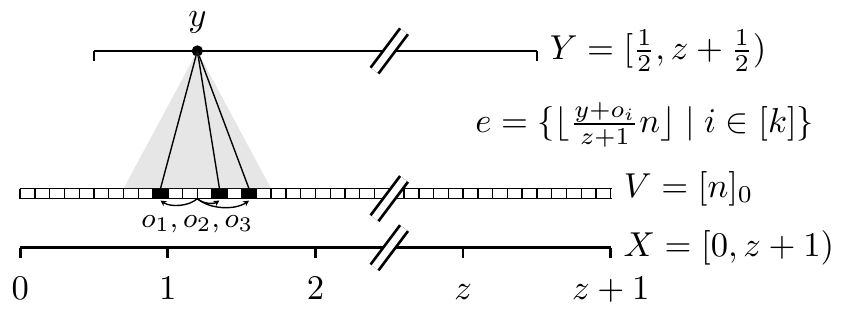}
    \caption[Construction of the spatially coupled hypergraph.]{In the construction from \cref{def:hypergraphs}, the vertex set $V = [n]₀$ is arranged linearly along the “coupling dimension” $X = [0,z+1)$. Thus each vertex has a position $x ∈ X$. Each hyperedge $e$ is independently obtained as follows. First, pick a random position $y$ uniformly from $Y = [\frac 12,z+\frac 12)$. Then pick the $k$ incidences of $e$ independently and uniformly at random from the vertices with positions in $[y-\frac 12,y+\frac 12]$.}
    \label{fig:coupled-hypergraph}
\end{figure}

It is possible (but rare) that incidences repeat within a hyperedge and that the same hyperedge appears several times in $F_n$. Note that the hyperedge density is ${|E|}/{|V|} = c\frac{z}{z+1}$ and only approaches $c$ for large $z$. The main technical contribution of this paper is the following Theorem.

\begin{theorem}
    \label{thm:main} Let $k,ℓ ∈ ℕ$, with $k ≥ 2$ and $k+ℓ ≥ 4$.  Then we have:
    \begin{enumerate}[(i)]
            • $∀c < \cklo\colon ∀z ∈ ℝ^+\colon \Pr[F_{n} \text { is $ℓ$-peelable}] \convn 1$.
            • $∀c > \cklo\colon ∃z^* ∈ ℝ^+\colon ∀z ≥ z^*\colon$\\\phantom{space}$\Pr[F_{n} \text { is $ℓ$-orientable}] \convn 0$.
    \end{enumerate}
\end{theorem}

Let us distil the main takeaways from these claims, including \cref{thm:condensed}.

\def\hF{\hat{F}}
\begin{corollary}
    \label{cor:coupled-peelable}
    Let $k,ℓ ∈ ℕ$ with $k ≥ 2$ and $k+ℓ ≥ 4$. For $z ∈ ℝ^+$ consider the family $(F(n,k,c,z))_{c ∈ ℝ₀^+, n ∈ ℕ}$. Let $\fklzp$ be its threshold for \emph{$ℓ$-peelability} and $\fklzo$ its threshold for $ℓ$-orientability. Then we have:
    \begin{enumerate}[(i)]
        • $∀z ∈ ℝ^+\colon \fklzp ≥ \cklo$.
        • $\limsup_{z → ∞} \fklzo ≤ \cklo$.
        • Let $\fklp = \lim_{z→∞} \fklzp$ and $\fklo = \lim_{z→∞} \fklzo$. Then $\fklp = \fklo = \cklo$.
    \end{enumerate}
\end{corollary}
\begin{proof}[Proof of \cref{cor:coupled-peelable}]
    Claims \ldl{(i)} and \ldl{(ii)} are immediate consequences of the claims from \cref{thm:main}. Since $\fklzp ≤ \fklzo$ we conclude \ldl{(iii)}.
\end{proof}
\begin{proof}[Proof of \cref{thm:condensed}]
    Let $c = \cklo - ε$ for some $ε > 0$. We define $\hF_{n,cn}^{k,ℓ} ≔ F(n,k,c\frac{z+1}{z},z = 2ℓ/ε)$. Note that this hypergraph is of the desired form with $n$ vertices and $\lfloor cn\rfloor$ i.i.d.\ hyperedges of size $k$.\footnote{We remark that the “diagonal” family $(\hF_{n,cn}^{k,ℓ})_{c,n}$ is mostly of theoretical interest. We do not claim that our choice of $z = 2ℓ/ε$ is particularly practical. The fact that the hyperedge density affects the underlying parameter $z$ is also inconvenient when constructing dynamic data structures.}
    
    Observe that $c \frac{z+1}{z} = c\frac{2ℓ/ε+1}{2ℓ/ε} = c(1+ε/(2ℓ)) ≤ \cklo - ε/2 ≤ \fklzp - ε/2$ using \cref{cor:coupled-peelable} \ldl{(i)} and the trivial bound $\cklo ≤ ℓ$. Therefore $\hF_{n,cn}^{k,ℓ}$ is $ℓ$-peelable with probability $1-o(1)$ by definition of $\fklzp$.
\end{proof}

Our construction is in the spirit of a technique from coding theory, see \cref{sec:spatial-coupling}.
Note that constructions similar to ours can already be found in \cite{HMU:Space-of-Solutions-XORSAT:2013} and \cite{GMU:Maxwell-Conjecture-Via-Spatial-Coupling:2012}, however, the goals of these papers are very different. Relative to these results, we can offer:
\ldl{(1)}~A generalisation to $ℓ > 1$.
\ldl{(2)}~A more elegant construction using the updated tools from \cite{KRU:Wave-Like:2015} (continuous\footnote{In \cite{HMU:Space-of-Solutions-XORSAT:2013} the coupling dimension is discrete. In our terms, this means that the set of admissible positions of a hyperedge is $Y ∩ (\frac 1w ℤ)$ for some constant $w ∈ ℕ$. Our construction arises for $w → ∞$.} coupling dimension).
\ldl{(3)}~A framing with data structures in mind and a demonstration of practical benefits for data structures.


\subsection{Comparison with Previous Constructions of Peelable Hypergraphs}
\label{sec:comparison}

\tikzstyle{resultPoint}=[circle,semithick,draw,inner sep=1]
\tikzstyle{rndOr}=[resultPoint,star,fill=yellow]
\tikzstyle{rndPl}=[resultPoint,diamond,fill=green]
\tikzstyle{lmss}=[resultPoint,rectangle,fill=red!60!black,draw=none,inner sep=1.5]
\tikzstyle{rink}=[resultPoint,fill=blue,inner sep=1]
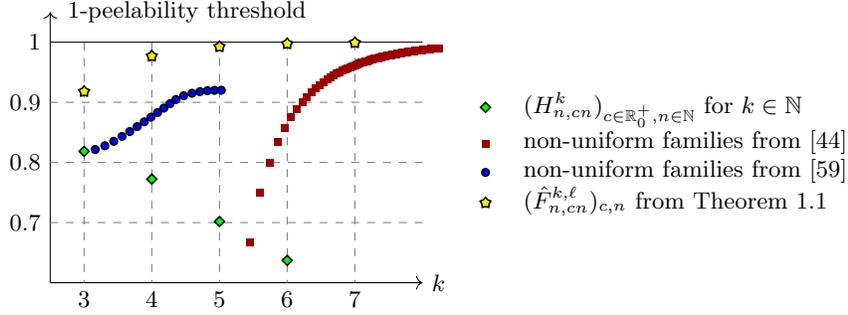
\begin{figure*}[t]
    \def\ymin{0.6}
    \def\xmin{2.5}
    \hbox{}~\hspace{-0.5cm}~\begin{tikzpicture}[xscale=0.9,yscale=8,font=\small]
        \draw[->] (\xmin,\ymin) -- (8,\ymin) node[right] {$k$};
        \draw[->] (\xmin,\ymin) -- (\xmin,1.05) node[right] {\ $1$-peelability threshold};
        \foreach \ytik in {0.7,0.8,0.9}{
            \draw[draw=gray,dashed] (8,\ytik) -- (\xmin,\ytik) node[left] {\ytik};
        }
        \foreach \xtik in {3,4,5,6,7}{
            \draw[draw=gray,dashed,thin] (\xtik,1) -- (\xtik,\ymin) node[below] {\xtik};
        }
        \draw (8,1) -- (\xmin,1) node[left] {1};
        \foreach \c[count=\k from 3] in {0.8184691607632809,0.7722798398025512,0.7017802664857019,0.6370811272741549,0.5817751769996049}{
            \ifnum\k<7
                \node[rndPl] at (\k,\c) {};
            \fi
        }
        \foreach \D/\avgDeg[count=\j] in {1/5,2/5.25,3/5.44444,4/5.60417,5/5.74,6/5.85833,7/5.96327,8/6.05759,9/6.1433,10/6.22187,11/6.29441,12/6.36181,13/6.42476,14/6.48382,15/6.53944,16/6.59202,17/6.64188,18/6.68928,19/6.73446,20/6.77763,21/6.81895,22/6.85858,23/6.89665,24/6.93329,25/6.9686,26/7.00267,27/7.03558,28/7.06743,29/7.09826,30/7.12815,31/7.15716,32/7.18532,33/7.2127,34/7.23933,35/7.26526,36/7.29052,37/7.31514,38/7.33916,39/7.36261,40/7.38551,41/7.40788,42/7.42976,43/7.45116,44/7.47211,45/7.49261,46/7.5127,47/7.53239,48/7.55169,49/7.57062,50/7.58919,51/7.60742,52/7.62531,53/7.64289,54/7.66016,55/7.67713,56/7.69382,57/7.71022,58/7.72636,59/7.74224,60/7.75787,61/7.77325,62/7.7884,63/7.80332,64/7.81801,65/7.8325,66/7.84677,67/7.86084,68/7.87471,69/7.88839,70/7.90188,71/7.91519,72/7.92832,73/7.94128,74/7.95408,75/7.96671,76/7.97918,77/7.99149,78/8.00366,79/8.01568,80/8.02755,81/8.03928,82/8.05087,83/8.06233,84/8.07366,85/8.08486,86/8.09594,87/8.10689,88/8.11773,89/8.12844,90/8.13904,91/8.14953,92/8.15991,93/8.17018,94/8.18035,95/8.19041,96/8.20038,97/8.21024,98/8.22,99/8.22967,100/8.23925}{
                \ifnum\j>2
                    \node[lmss] at (\avgDeg,1-1/\D) {};
                \fi
        }
        \foreach \bark/\c in {3.16404/0.82151,3.30658/0.82770,3.43744/0.83520,3.55944/0.84321,3.67439/0.85138,3.78359/0.85952,3.88795/0.86752,3.98818/0.87535,4.08482/0.88298,4.17830/0.89040,4.26898/0.89761,4.35715/0.90461,4.47102/0.91089,4.59372/0.91510,4.71015/0.91772,4.82077/0.91922,4.92601/0.91992,5.02626/0.92004}{
            \node[rink] at (\bark,\c) {};
        }
        \foreach \c[count=\k from 3] in {0.917935,0.97677,0.99243,0.99738,0.99906}{
            \node[rndOr] (new-\k) at (\k,\c) {};
        }
        
        \node[right,align=left] at (current bounding box.east) {
            \begin{tabular}{cl}
                \tikz \node[rndPl]{}; & $\FRfamily$ for $k ∈ ℕ$ \\ 
                \tikz \node[lmss]{}; & non-uniform families from \cite{LMSS:Efficient_Erasure:2001}  \\
                \tikz \node[rink]{}; & non-uniform families from \cite{R:Mixed:2013} \\
                \tikz \node[rndOr]{};& $(\hF_{n,cn}^{k,ℓ})_{c,n}$ from \cref{thm:condensed}\\
            \end{tabular}
        };
    \end{tikzpicture}
    \caption[Comparison of peelable hypergraph constructions.]{Trade-offs between hyperedge size and $1$-peelability threshold.\\
    Concretely, a dot at $(k,\cpnoindex) ∈ ℝ²$ indicates a family \smash{$(H_{c,n})_{c ∈ ℝ₀^+,n ∈ ℕ}$} of random hypergraphs where $H_{c,n}$ has $n$ vertices, $\lfloor cn\rfloor$ random independent hyperedges, expected hyperedge size $k$. The value $\cpnoindex$ is the $1$-peelability threshold of the family. 
    }
    \label{fig:comparison}
    \vspace{1em}
    \hrule
\end{figure*}

The applications in \cref{sec:HBDS} require hypergraph families with i.i.d.\ random hyperedges, preferably with small average hyperedge size $k$ and large $ℓ$-peelability threshold.
In \cref{fig:comparison}, we compare previous constructions to our own with respect to these criteria in the most relevant case of $ℓ = 1$.

The thresholds $\cp{k}{1}$ of the fully random families $\FRfamily$ for $k ≥ 3$ (\tikz \node[rndPl] {};) \cite{Molloy05:Cores-in-random-hypergraphs} are decreasing in $k$ and thus mostly $k = 3$ is of interest.
The thresholds of a non-uniform construction (\raisebox{1.5pt}{\tikz \node[lmss] {};}) \cite{LMSS:Efficient_Erasure:2001}, well-known in coding theory, approach $1$ for $k → ∞$. However, the maximum hyperedge size is exponential in the average hyperedge size $k$, which is problematic for some applications. Further trade-offs (\tikz \node[rink] {};) were examined by \cite{R:Mixed:2013}, for example, a hyperedge size of $3$ for $≈ 89\%$ of the hyperedges and a size of $21$ for the rest in an otherwise fully random construction yields an average hyperedge size of $≈ 5.03$ and a threshold value of $≈ 0.92$.
For $k ≥ 3$, the threshold of the family $(\hF_{n,cn}^{k,1})_{c ∈ [0,\co{k}{1}), n ∈ ℕ}$ proposed in \cref{thm:condensed} (\tikz \node[rndOr] {};) is the \emph{$1$-orientability} threshold $\co{k}{1}$ of $\smash{\FRfamily}$.

\subsection{The Technique of Spatial Coupling}
\label{sec:spatial-coupling}

The hypergraph $F_n$ arises by “spatial coupling” of $H_{n,cn}^k$ along the “coupling dimension” $X = [0,z+1)$, which roughly means that $F_n$ inherits its \emph{local} structure from $H_{n,cn}^k$ (cf. \cref{fact:fully-random-rwl,lem:limits-coincide}) but its global structure reflects $X$. Consider the parallel peeling process on $F_n$ that deletes in each of its rounds all vertices of degree at most $ℓ$ (and incident hyperedges). Vertices with a position close to the borders $0$ or $z+1$ tend to be deleted early on, while many vertices in the denser, central parts remain. But gradually, deletions at the border “expose” vertices further on the inside and the whole hypergraph “erodes” from the outside in. The effect is shown in \cref{fig:couplingPhenomenon}. This does not happen in the more symmetric construction when $X$ is glued into a circle, i.e.\ if for all $ε ∈ [0,1)$ the positions $ε$ and $z+ε$ are identified.

\begin{figure*}[t]
    \includegraphics{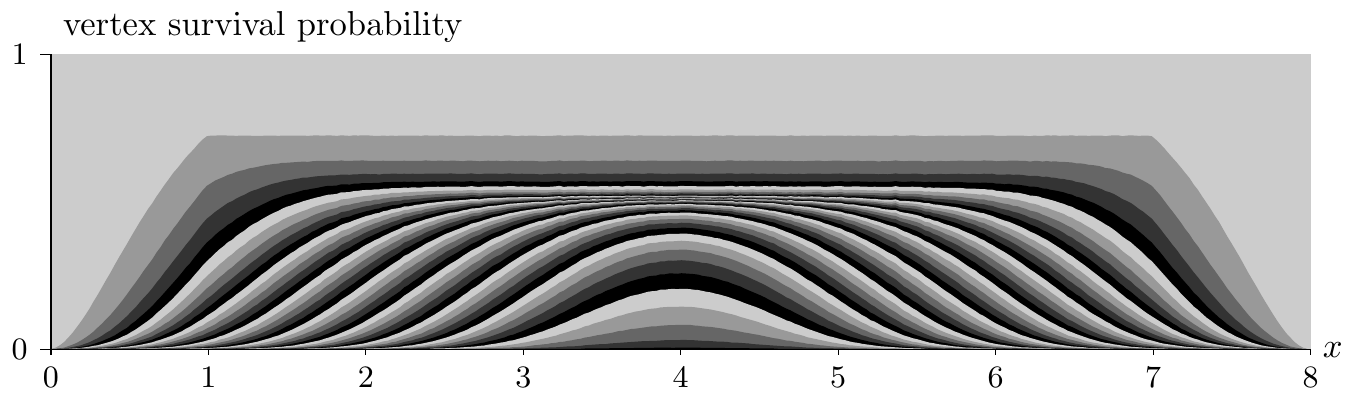}
    \caption[“Layers” of a spatially coupled hypergraph during $48$ rounds of peeling.]{The “layers” of a hypergraph sampled according to \cref{def:hypergraphs} with parameters $k =3$, $c = 0.85$, $z = 7$ and $m = 10⁸$.
    It happens to be peelable in $48$ rounds, that are assigned shades of gray in a round-robin fashion. For each round $r ∈ \{0,…,48\}$ and each position $x ∈ X = [0,8]$, the picture shows the fraction of vertices close to $x$ that “survive” $r$ rounds of the parallel peeling process, a value related to $q^{(r)}(x) ∈ [0,1]$ (see \cref{sec:peeling-operator}).
    }
    \vspace{1em}
    \hrule
    \label{fig:couplingPhenomenon}
\end{figure*}

The authors of \cite{KMSSZ:Statistical-Physics-in-Compressed-Sensing:2011,KRU:Coupled-Achieves-Capacity:2013} liken the phenomenon to water that is \emph{super-cooled} to below $0$°C in a smooth container. It will not freeze unless a \emph{nucleus} for crystallization is introduced. Once this is done, all water crystallises quickly, starting from that nucleus. In this sense, the introduction of a nucleus raises the threshold for crystallisation. 

In our construction, the borders play the role of such a nucleus and raise the peeling threshold to the orientability threshold. Similarly, in coding theory the threshold for decoding with a greedy (belief propagation) decoder is raised to the threshold for decoding with the ideal (maximum aposteriori probability) decoder. This effect is known as \emph{threshold saturation}.


We leave a summary of the field to the experts \cite{KRU:Wave-Like:2015,KRU:Coupled-Achieves-Capacity:2013}. Put briefly, the phenomenon was discovered in the form of convolutional codes \cite{FZ:Convolutional-LDPC:99}, then rigorously explained, first in a special case \cite{KRU:Why-Conv-LDPC-Works:2010}, then more generally \cite{KRU:Coupled-Achieves-Capacity:2013}, later accounting for continuous coupling dimensions (and even multiple dimensions) \cite{KRU:Wave-Like:2015}, a form we will exploit in this paper.

\subsection{Outline}

The proof is organised as follows. In \cref{sec:peeling-operator} we idealise the peeling process by switching to a tree-like distributional limit of our hypergraphs, and capture the essential behaviour of the process in terms of an operator $\hP$ acting on functions $q\colon ℝ → [0,1]$. In \cref{sec:interpolating-wave} we analyse the effect of iterated application of $\hP$ to functions using the rich toolbox from \cite{KRU:Wave-Like:2015}. This is the main ingredient to proving part \ldl{(i)} of \cref{thm:main} in \cref{sec:wrapping-up}. The comparatively simple part \ldl{(ii)} is independent of these considerations and is proved in \cref{sec:upperbound}.

Finally, in \cref{sec:experiments} we demonstrate how using our hypergraphs can improve the performance of practical retrieval data structures.

%% file: peeling-process.tex

\section{The Peeling Process and Idealised Peeling Operators}
\label{sec:peeling-operator}

This section examines how the probabilities for vertices of $F_n$ to “survive” $r ∈ ℕ$ rounds of peeling change from one round to the next. In the classical setting, this could be described by a function, mapping the old survival probability to the new one \cite{Molloy05:Cores-in-random-hypergraphs}. In our case, however, there are distinct survival probabilities $q(x)$ depending on the position $x$ of the vertex. Thus we need a corresponding operator $\hP$ that acts on such functions~$q$. 

We almost always suppress $k,ℓ,c,z$ in notation outside of definitions. Big-$\O$ notation refers to $n → ∞$ while $k,ℓ,c,z$ are constant.

Consider the parallel peeling process $\peel(F_n,ℓ)$ on $F_n = F(n,k,c,z)$. In each \emph{round} of $\peel(F_n)$, all vertices of degree at most $ℓ$ are determined and then deleted simultaneously. Deleting a vertex implicitly deletes all incident hyperedges. We also define the \emph{$r$-round rooted peeling process} $\peel_{v,r}(F_n,ℓ)$ for any vertex $v ∈ V$ and $r ∈ ℕ$. In round $1 ≤ r' ≤ r-1$ of $\peel_{v,r}(F_n)$, only vertices with distance $r-r'$ from $v$ are considered for deletion. Moreover, in round $r$, the root vertex $v$ is only deleted if it has degree at most $ℓ-1$, not if it has degree~$ℓ$.

For any vertex position $x ∈ X = [0,z+1)$ and $r ∈ ℕ$ we let $\q{r}(x) = \q{r}(x,n,k,ℓ,c,z)$ be the probability that the vertex $v = \lfloor \frac{x}{z+1}n\rfloor$ survives $\peel_{v,r}(F_n)$, i.e.~is not deleted. It is convenient to define $\q{0}(x) = 1$ for all $x ∈ X$, i.e.~every vertex survives the “0-round peeling process”. To get an intuition for how $\q{r}(x)$ evolves with $r$, consider \cref{fig:couplingPhenomenon}   Even though $\q{r}$ is discrete in $x$ by definition, we will later see that it has a continuous limit for $n → ∞$.

 
Whether a vertex $v$ at position $x$ survives $\peel_{v,r}$ is a function of its $r$-neighbourhood $F_n(x,r)$, i.e.~the subhypergraph of $F_n$ that can be reached from $v$ by traversing at most $r$ hyperedges.

It is natural to consider the distributional limit of $F_n(x,r)$ to get a grip on $\q{r}(x)$. In the spirit of the objective method \cite{AS:Objective_Method:2004}, we identify a (possibly infinite) random tree $T_x$ that captures the local characteristics of $F_n(x,r)$ for $n → ∞$. In the following $\Po(λ)$ refers to the Poisson distribution with mean $λ ∈ ℝ^+$.
 
 

 \begin{Definition}[Limiting Tree]
    \label{def:coupling-T}
    Let $k ∈ ℕ$, $c, z ∈ ℝ^+, X = [0,z+1), Y = [\frac 12,z + \frac 12)$ and $x ∈ X$. The random (possibly infinite) hypertree $T_x = T_x(k,c,z)$ is distributed as follows.
    
    $T_x$ has a root vertex $\root(T_x)$ at position $x$, which for $Y_x ≔ [x-\frac 12,x+\frac 12] ∩ Y$ has $d_x \sim \Po(ck|Y_x|)$ \emph{child hyperedges} with positions uniformly distributed in $Y_x$.\footnote{In other words: The positions of the child hyperedges are a Poisson point field on $Y_x$ with intensity $ck$. By $|I|$ for an interval $I = [a,b]$ we mean $b-a$.\\    
    Note also that the position is now a property of a vertex, not an identifying feature. Possibly (though with probability $0$) the tree $T_x$ may contain several vertices with the same position.
    } Each child hyperedge at position $y$ is incident to $k-1$ (fresh) \emph{child vertices} of its own, each with a uniformly random position $x' ∈ [y-\frac 12,y+\frac 12]$. The sub-hypertree at such a child vertex at position $x'$ is distributed recursively (and independently of its sibling-subtrees) according to~$T_{x'}$.
 \end{Definition}

For $x ∈ X$ and $r ∈ ℕ$, let $F_n(x,r)$ and $T_x(r)$ denote the $r$-neighbourhoods of vertex $v = \lfloor \frac{x}{z+1}n\rfloor$ in $F_n$ and $\root(T_x)$ in $T_x$, respectively. In the following, $H$ is an arbitrary fixed rooted hypergraph and equality of hypergraphs indicates a root-preserving isomorphism.

\begin{lemma}
    \label{lem:localConvergenceOfH}
    $\displaystyle∀x ∈ X,r ∈ ℕ,H\colon$\\\phantom{space}$\lim_{n→∞} \Pr[F_n(x,r) = H] = \Pr[T_x(r) = H]$.
\end{lemma}

\begin{proof}[Sketch of Proof.]
    We construct for fixed $r,x$ and $H$ a random coupling\footnote{There is no relation to the term spatial coupling. We refer to the standard technique where several random variables are realised on the same probability space.} between $F_n(x,r)$ and $T_x(r)$ such that the symmetric difference between the events $\{F_n(x,r) = H\}$ and $\{T_x(r) = H\}$ has probability $o(1)$.
    We do so inductively, by following a sequence of events. The $i$-th event expresses, firstly, that $F_n(x,r)$ and $T_x(r)$ agree with $H$ concerning the first $i$ rounds of a breadth-first search traversal and, secondly, that corresponding active vertices in $F_n(x,r)$ and $T_x(r)$ have positions with distance $\O(1/n)$.
    
    For the first step, consider the root $v$ of $F_n(x,r)$. By construction, any hyperedge containing $v$ must have a position $y ∈ [x-\frac 12, x+\frac 12]$. For $x ∈ [0,1)$ or $x ∈ [z,z+1)$ the potential positions are further restricted by the upper and lower bounds on hyperedge positions, i.e.~we have $y ∈ Y_x ≔ [x-\frac 12,x+\frac 12] ∩ Y$. In order for a random hyperedge $e$ to contain $v$, two things have to work out:
    \begin{enumerate}[(1)]
            • The position of $e$ must fall within $Y_x$. The probability for this is $|Y_x|/|Y| = |Y_x|/z$.
            • At least one of the $k$ incidences of $e$ must turn out to be to $v$. The probability for this is $1-(1-\frac{z+1}{n})^k$.
    \end{enumerate}
    With $cn\frac{z}{z+1}$ hyperedges in total, we obtain a binomial distribution $\deg(v) \sim \Bin\big(cn\frac{z}{z+1},\allowbreak|Y_x|/z(1-(1-\frac{z+1}n)^k)\big)$. This distribution converges, for $n → ∞$, to $\Po(ck|Y_x|)$, which is the distribution of $\deg(\root(T_x))$. The positions of the neighbours of $v$ are uniformly distributed in the discrete set $|Y_x| ∩ (\frac{z+1}{n}ℤ)$, the positions of the neighbours of $\root(T_x)$ uniformly in the interval $|Y_x|$. It should now be easy to see how a coupling between $F_n(x,1)$ and $T_x(1)$ could look like.
    
    There are three complications when continuing the argument: \ldl{(i)} The discrepancies between vertex positions of $F_n(x,r)$ and $T_x(1)$ need to be kept in check. \ldl{(ii)} $F_n(x,r)$ may contain cycles\footnote{This can actually already occur for $r = 1$.}. \ldl{(iii)} There are slight dependencies between vertex degrees in $F_n(x,r)$.\ \ It should be intuitively plausible that these problems vanish in the limit. We refer to \cite{Leconte:Cuckoo:2013} for a full argument     showing a similar convergence and to \cite{K:Poisson:2006} for the related technique of Poissonisation.
\end{proof}
  
\noindent We now consider the idealised peeling processes $(\peel_{\root(T_x),r}(T_x))_{x ∈ X}$. Their survival probabilities are easier to analyse than those of $\peel_{v,r}(F_n)$.
 
 \begin{lemma}
    \label{lem:peelingT}
    Let $r ∈ ℕ₀$ be constant and $\q{r}_T(x) = \q{r}_T(x,k,ℓ,c,z)$ be the probability that $\root(T_x)$ survives $\peel_{\root(T_x),r}(T_x,ℓ)$ for $x ∈ X$. Then for $x ∈ X$
    \[ \q{r+1}_T(x) = Q\Bigg(ck
        \!\!\!\!\!\! \int\limits_{[x-\frac 12,x+\frac12]∩Y} \!\!\!\!\!\!
        \bigg( \int_{y-\frac 12}^{y+\frac 12}\q{r}_T(x') dx' \bigg)^{k-1} dy,\ \ ℓ\Bigg).\]
    where $Q(λ,ℓ) = 1- \sum_{i < ℓ} \frac{λ^i}{i!} = \Pr[\Po(λ) ≥ ℓ]$, the latter term slightly abusing notation.
 \end{lemma}
 
 \begin{proof}
    Let $x ∈ X$ and $v = \root(T_x)$. Assume $y ∈ [x-\frac 12,x+\frac 12] ∩ Y$ is the type of some hyperedge $e$ incident to $v$. Hyperedge $e$ survives $r$ rounds of $\peel_{v,r+1}(T_x)$ if and only if all of its incident vertices survive these $r$ rounds. Since $v$ itself may only be deleted in round $r+1$, the relevant vertices are the $k-1$ child vertices $w₁,…,w_{k-1}$ with positions uniformly distributed in $[y-\frac 12,y+\frac 12]$. Let $W_i$ be the subtree rooted at $w_i$ for $1 ≤ i < k$. Consider the peeling process $\peel_{w_i,r}(W_i)$. Assume the process deletes $w_i$ in round $r$, meaning $w_i$ has degree at most $ℓ-1$ at the start of round $r$. Then $w_i$ has degree at most $ℓ$ at the start of round $r$ in $\peel_{v,r+1}(T_x)$, meaning $\peel_{v,r+1}(T_x)$ deletes $e$ in round $r$. Conversely, if none of $\peel_{w₁,r}(W₁),…,\peel_{w_{k-1},r}(W_{k-1})$ delete their root vertex within $r$ rounds, then $w₁,…,w_{k-1}$ have degree at least $ℓ+1$ after round $r$ of $\peel_{v,r+1}(T_x)$ and $e$ survives round $r$ of $\peel_{v,r+1}(T_x)$. Since the position of each $w_i$ is independent and uniformly distributed in $[y-\frac 12,y+\frac 12)$, the probability for $e$ to survive is  $p_{y} ≔ (\int_{y-\frac 12}^{y+\frac 12}\q{r}_T(x') dx' )^{k-1}$. Since the positions of the hyperedges incident to $v$ are a Poisson point field on $[x-\frac 12,x+\frac 12] ∩ Y$ with intensity $ck$, the number of incident hyperedges \emph{surviving} round $r$ of $\peel_{v,r+1}(T_x)$ has Poisson distribution with mean $λ ≔ \int_{[x-\frac 12,x+\frac 12] ∩ Y} ckp_{y} dy$.
    
    The claim now follows by observing that $v$ survives $r+1$ rounds of $\peel_{v,r+1}(T_x)$ if it is incident to at least $ℓ$ hyperedges surviving $r$ rounds. The probability for this is $Q(λ,ℓ)$.
 \end{proof}
 %
For convenience, we define the operator ${\P} ={\P}(k,ℓ,c,z)$, which maps any (measurable\footnote{All functions that play a role in our analysis are measurable. We refrain from pointing this out from now on.}) $q \colon X → [0,1]$ to $\P q \colon X → [0,1]$ with 
 \[ (\P q)(x) = Q\Bigg(ck
 \!\!\!\!\!\int\limits_{[x-\frac 12,x+\frac 12]∩Y}\!\!\!\!\!
 \bigg( \int_{y-\frac 12}^{y+\frac 12} q(x') dx' \bigg)^{k-1} dy,\ \ ℓ\Bigg).\]
 Together Lemmas \ref{lem:localConvergenceOfH} and \ref{lem:peelingT} imply that $\P$ can be used to approximate survival probabilities.
 \def\err{\mathrm{err}}
 \begin{corollary}
    \label{cor:approxOfq}
    Let $r ∈ ℕ₀$ be constant. Then for all $x ∈ X$
    \begin{align*}
        \P^r \q{0}(x) &\stackrel{\textup{def}}= \P^r \q{0}_T(x)\\
        &\!\!\!\!\stackrel{\textup{Lem\,\ref{lem:peelingT}}}= \q{r}_T(x) \stackrel{\textup{Lem\,\ref{lem:localConvergenceOfH}}}= \q{r}(x)\pm o(1).
    \end{align*}
 \end{corollary}
\noindent To obtain \emph{upper} bounds on survival probabilities, we may remove the awkward restriction “$∩\,Y$” in the definition of $\P$. We define $\hP = \hP(k,ℓ,c)$ as mapping any $q \colon ℝ → [0,1]$ to $\hP q \colon ℝ → [0,1]$ with
\[ (\hP q)(x) = Q\Bigg(ck \int_{x-\frac 12}^{x+\frac 12} \bigg( \int_{y-\frac 12}^{y+\frac 12} q(x') dx' \bigg)^{k-1} dy,\ \ ℓ\Bigg)\]
Note that $\hP$ does not depend on $z$ or $n$. To simplify notation, we assume that the old operator $\P$ also acts on functions $q \colon ℝ → [0,1]$, ignoring $q(x)$ for $x ∉ X$, and producing $\P q \colon ℝ → [0,1]$ with $\P q(x) = 0$ for $x ∉ X$. We also extend $\q{0}$ to be $𝟙[x ∈ {X}] \colon ℝ → [0,1]$, i.e.~the characteristic function on $X$, essentially introducing vertices at positions $x ∉ X$ which are, however, already deleted with probability $1$ before the first round begins.
Note that while $\q{r}(x)$ and $\q{r}_T(x)$ are by definition non-increasing in $r$, this is not the case for $(\hP^r \q{0})(x)$. For instance, $\hP^{r} \q{0}$ has support $(-r,z+1+r)$, which grows with $r$.\footnote{It is still possible to interpret $\hP^r \q{0}(x)$ as survival probabilities in more symmetric, extended versions $\hat{T}_x$ of the tree $T_x$, but we will not pursue this.}
The following lemma lists a few easily verified properties of $\hP$. All inequalities between functions should be interpreted point-wise.
\begin{lemma}
    \label{lem:properties-hP}
    \begin{enumerate}[(i)]
            • $∀q\colon ℝ → [0,1]\colon \P q ≤ \hP q$.
            • $\P$ and $\hP$ are monotonic, i.e.~$∀q,q'\colon ℝ → [0,1] \colon$\\
            $q ≤ q' ⇒ \P q ≤ \P q' ∧ \hP q ≤ \hP q'$.
            • $\P$ and $\hP$ are continuous, i.e.~pointwise convergence of $(q_i)_{i ∈ ℕ}$ to $q^*$ implies pointwise convergence of $(\P q_i)_{i ∈ ℕ}$ and $(\hP q_i)_{i ∈ ℕ}$ to $\P q^*$ and $\hP q^*$, respectively.
    \end{enumerate}
\end{lemma}

%% file: applying-richardson-urbanke.tex

\section{Analysis of Iterated Peeling}
\label{sec:interpolating-wave}

The goal of this section is to prove the following proposition.

\begin{proposition}\ 
    \label{prop:convergence-of-q}
    \begin{enumerate}[(i)]
        • For $c < c_{k,ℓ}^*$ and any $z ∈ ℝ^+$, we have $(\P^r q₀)(x) \conv 0$ for all $x ∈ X$.
        • For $c > c_{k,ℓ}^*$ and large $z$, we have $(\P^r q₀)(x) \conv q^*(x)$ for all $x ∈ X$ and some $q^* ≠ 0$.
    \end{enumerate}
\end{proposition}
The intuition is that for $c > c_{k,ℓ}^*$ the peeling process gets stuck, while for $c < c_{k,ℓ}^*$ all vertices are eventually peeled.

Conveniently, iterations such as the one given by $\P$ and $\hP$ were extensively studied in a stunning paper by Kudekar, Richardson and Urbanke \cite{KRU:Wave-Like:2015}. For some initial function $f^{(0)} \colon ℝ → [0,1]$ and non-decreasing functions $h_f, h_g \colon [0,1] → [0,1]$ they study the sequence of functions
\begin{align}
    g^{(r)}(y) &≔ h_g((f^{(r)} ⊗ ω)(y))\label{eq:coupledSystem}\\
    f^{(r+1)}(x) &≔ h_f((g^{(r)} ⊗ ω)(x))\notag
\end{align}
where $ω$ is an \emph{averaging kernel}, i.e.~an even non-negative function with integral $1$ and $⊗$ is the convolution operator. To apply the theory to our case, we use:
\begin{align*}
    h_f(u) &≔ Q(cku,ℓ) \\ h_g(v) &≔ v^{k-1} \\ ω(x) &= 𝟙[|x| ≤ \frac 12]
\end{align*}
With these substitutions the iteration (\ref{eq:coupledSystem}) satisfies $\hP f^{(r)} = f^{(r+1)}$. If we force the functions $g^{(r)}$, $r ∈ ℕ$, to be zero outside of $Y = [\frac 12,z+\frac 12)$ by replacing (\ref{eq:coupledSystem}) with $g^{(r)}(y) ≔ \min\{𝟙[y ∈ Y], h_g((f^{(r)} ⊗ ω)(y))\}$ we get the system with \emph{two-sided termination}. In this case $\P f^{(r)} = f^{(r+1)}$. The system with \emph{one-sided termination} is defined similarly with $Y = [\frac 12,∞)$.

We remark that nothing in the following depends on the choice of $ω$.\footnote{There is a corresponding flexibility in \cref{def:hypergraphs}. Instead of a hyperedge at position $y$ choosing its incident vertices uniformly at random from $[y - \frac 12, y + \frac 12]$, incidences can be chosen according to an almost arbitrary bounded density function that is symmetric around $y$. For details consider \cite[Definition 2]{KRU:Wave-Like:2015}.}

\subsection{Unleashing Heavy Machinery from Coding Theory}

We plan to delegate the proof of \cref{prop:convergence-of-q} to theorems from \cite{KRU:Wave-Like:2015}. For this, we need to examine the \emph{potential} $ϕ(u,v) = ϕ(h_f,h_g,u,v)$ given as:
\begin{gather*}
    ϕ(u,v) = \int_{0}^u h_g^{-1}(u') du' + \int_0^v h_f^{-1}(v')dv' - uv \\
    \text{ for $0 ≤ u ≤ h_g(1),$\  $0 ≤ v ≤ h_f(1)$.}
\end{gather*}
\begin{figure}
    \includegraphics{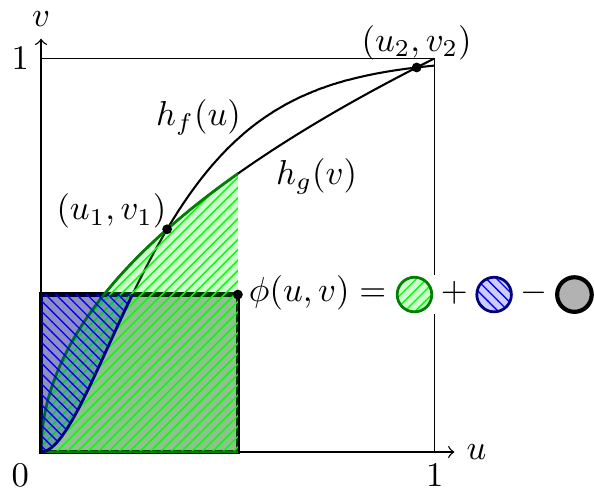}
    \caption[Visualisation of the potential function $ϕ$.]{A plot of the curves $u ↦ (u,h_f(u))$ and $v ↦ (h_g(v),v)$ for $u,v ∈ [0,1]$ with $k = 3$, $ℓ = 2$ and $c = c_{k,ℓ}^*$. The three crossing points of the curves are the solutions $(0,0)$, $(u₁,v₁)$ and $(u₂,v₂)$ to \cref{eq:fixed-point}. The potential $ϕ(u,v)$ can be visualised as the sum of three areas as shown. The significance of the threshold $c_{k,ℓ}^*$ is that the two areas enclosed by the two curves have exactly the same size, or put differently, $ϕ(u₂,v₂) = 0$.}
    \label{fig:potentialPic}
\end{figure}
A visualisation is given in \cref{fig:potentialPic}. Consider the equation
\begin{equation}
    (u,v) = (h_g(v),h_f(u)).\label{eq:fixed-point}
\end{equation}
Clearly it has the \emph{trivial solution} $(u,v) = (0,0)$. By monotonicity of $h_g$ and $h_f$, any two solutions $(u₁,v₁)$ and $(u₂,v₂)$ are component-wise ordered. We write $(u₁,v₁) < (u₂,v₂)$ for $u₁ < u₂ ∧ v₁ < v₂$.
\begin{lemma}
    \label{lem:analysis-of-f-g}
    \begin{enumerate}[(i)]
        • \label{it:extremum-is-sol} Every local minimum $(u,v)$ of $ϕ$ is a solution to \cref{eq:fixed-point}.
        • \label{it:smallest-sol-pos} If \cref{eq:fixed-point} has at least one non-trivial solution, then the smallest non-trivial solution $(u₁,v₁)$ has potential $ϕ(u₁,v₁) > 0$.
        • \label{it:two-sols} \cref{eq:fixed-point} has at most two non-trivial solutions.
        • \label{it:touching-zero} For $c = c_{k,ℓ}^*$ there is a non-trivial solution $(u₂,v₂)$ of \cref{eq:fixed-point} with $ϕ(u₂,v₂) = 0$. In this case, $(0,0)$ and $(u₂,v₂)$ are the only minima of $ϕ$.
        • \label{it:positive} For $c < c_{k,ℓ}^*$ we have $ϕ(u,v) > 0$ for $(u,v) ≠ (0,0)$.
        • \label{it:non-positive} For $c > c_{k,ℓ}^*$ Equation (\ref{eq:fixed-point}) has two non-trivial solutions $(u₁,v₁) < (u₂,v₂)$. They satisfy $ϕ(u₂,v₂) < ϕ(0,0) = 0 < ϕ(u₁,v₁)$.
    \end{enumerate}
\end{lemma}

\begin{proof}
    \begin{enumerate}[(i)]
    • The partial derivatives of $ϕ$ are $\nabla ϕ(u,v) = (h_g^{-1}(u)-v,h_f^{-1}(v)-u)$. Therefore, the only candidates for local minima of $ϕ$ are the solutions to Equation (\ref{eq:fixed-point}) (it is easy to check that, except for $(u,v) = (0,0)$, there are no local minima at the borders).
    • Assume $(u₁,v₁)$ is the smallest non-trivial solution to \cref{eq:fixed-point}. Considering \cref{fig:potentialPic}, we see that $|ϕ(u₁,v₁)|$ is the area enclosed by $h_f(u)$ and $h_g^{-1}(u)$ for $u ∈ [0,u₁]$. To see that the sign of $ϕ(u₁,v₁)$ is positive, observe that for small values of $u$ we have $h_f(u) = Q(cku,ℓ) = \O(u^ℓ)$ while $h_g^{-1}(u) = Ω(u^{1/(k-1)})$ and thus $h_f(ε) < h_g^{-1}(ε)$ for $ε ∈ (0,u₁)$. This uses $ℓ ≥ 1$, $k ≥ 2$ and $(k,ℓ) ≠ (2,1)$.
    • By expanding $h_f$ and $h_g$ and substituting $ξ = ckv^{k-1}$ we get for $v ≠ 0$:
    \begin{gather*}
        (u,v) = (h_g(v),h_f(u))\\
        ⇒ v = Q(ckv^{k-1},ℓ)\\
        ⇔ \tfrac{ξ}{ck} = Q(ξ,ℓ)^{k-1}\\
        ⇔ \tfrac{ξ}{Q(ξ,ℓ)^{k-1}} = ck.
    \end{gather*}
    To show that the right-most equation has at most two solutions it suffices to show that $\frac{ξ}{Q(ξ,ℓ)^{k-1}}$ has at most one local extremum. If $ξ$ is such an extremum, we get
    \begin{align*}
        &\tfrac{d}{dξ}\tfrac{ξ}{Q(ξ,ℓ)^{k-1}} = 0\\
        ⇒&\ Q(ξ,ℓ)^{k-1} - ξ(k-1)Q(ξ,ℓ)^{k-2}Q'(ξ,ℓ) = 0\\
        ⇒&\ Q(ξ,ℓ) - ξ(k-1)Q'(ξ,ℓ) = 0\\
        ⇒&\sum_{i ≥ ℓ}\tfrac{ξ^i}{i!} - (k-1)ξ^{ℓ}{(ℓ-1)!} = 0\\
        ⇒&\sum_{i ≥ 0}\tfrac{ξ^i}{(i+ℓ)!} = (k-1){(ℓ-1)!}
    \end{align*}
    Since the left hand side is increasing in $ξ$ for $ξ > 0$ while the right hand side is constant, there is exactly one solution $ξ$ as claimed.
    • Recall that $c$ occurs in the definition of $h_f$ and note that $ϕ$ is monotonically decreasing in $c$. It is easy to see that $ϕ$ is nowhere negative for small values of $c$, and negative for some $(u,v)$ if $c$ is large. For continuity reasons and because $ϕ(u,v) ≥ 0$ for $u,v ∈ [0,ε]$ with $ε = ε(c)$ small enough (using similar arguments as in \itRef{it:smallest-sol-pos}), there must be some intermediate value $c$ where $ϕ(u₂,v₂) = 0$ for a local minimum $(u₂,v₂) ≠ (0,0)$ of $ϕ$. By \itRef{it:extremum-is-sol}, $(u₂,v₂)$ is a solution of \cref{eq:fixed-point}. By \itRef{it:smallest-sol-pos} there must be a smaller solution $(u₁,v₁)$ with $ϕ(u₁,v₁) > 0$. Now by \itRef{it:extremum-is-sol},  and \itRef{it:two-sols}, there cannot be minima of $ϕ$ in addition to $(0,0)$ and $(u₂,v₂)$. The only thing left to show is $c = c_{k,ℓ}^*$.
    
    We rewrite the potential at $(u₂,v₂)$, using \cref{eq:fixed-point}
    \begin{align*}
        ϕ(u₂,v₂) &= \int_{0}^{u₂} h_g^{-1}(u) du + \int_0^{v₂} h_f^{-1}(v)dv - u₂v₂\\
        &= \Big(u₂v₂ - \int_{0}^{v₂} h_g(v) dv\Big)\\
        &\phantom{=}+ \Big(u₂v₂ - \int_0^{u₂} h_f(u)du\Big) - u₂v₂\\
        &= v₂h_g(v₂) - H_g(v₂) - H_f(h_g(v₂)),
    \end{align*}
    where $H_g$ and $H_f$ are antiderivatives of $h_g$ and $h_f$, i.e:
    \begin{align*}
        H_g(v) &= \int h_g(v) dv = \tfrac{1}{k}v^{k} \\ 
        H_f(u) &= \int h_f(u) du = u - \tfrac{1}{ck}\sum_{i = 1}^{ℓ} Q(cku,i).
    \end{align*}
    The fact that $\int_{0}^λ Q(x,ℓ) dx = λ - \sum_{i = 1}^{ℓ} Q(λ,i)$ can be seen by induction on $ℓ$. We now examine the implications of $ϕ(u₂,v₂) = 0$.
    
    In the following calculation let $ξ ≔ ckv₂^{k-1}$ which implies $Q(ξ,ℓ) = v₂$.
    {\allowdisplaybreaks
    \begin{align*}
        0 &= ϕ(u₂,v₂) = v₂h_g(v₂) - H_g(v₂) - H_f(h_g(v₂))\\
         &= v₂^k - v₂^k/k - v₂^{k-1} + \tfrac{1}{ck}\sum_{i = 1}^{ℓ} Q(ckv₂^{k-1},i)\\
        ⇒ 0 &= ξv₂ - ξv₂/k - ξ + \sum_{i=1}^{ℓ} Q(ξ,i)\\[-6pt]
        &= ξQ(ξ,ℓ) - ξQ(ξ,ℓ)/k - ξ + \sum_{i=1}^{ℓ} Q(ξ,i)\\[-3pt]
        ⇒ \phantom{0}&ξQ(ξ,ℓ)/k = ξ(Q(ξ,ℓ)-1) + \sum_{i=1}^{ℓ} Q(ξ,i)\\
        &= -\e^{-ξ}\sum_{j = 0}^{ℓ-1}\frac{ξ^{j+1}}{j!} + \sum_{i=1}^{ℓ}\big(1 -\e^{-ξ}\sum_{j=0}^{i-1} \frac{ξ^{j}}{j!}\big)\\
        &= ℓ - e^{-ξ}\sum_{j=0}^{ℓ-1} (\frac{ξ^{j+1}}{j!} + (ℓ-j)\frac{ξ^j}{j!})\\ 
        &= ℓ - e^{-ξ}\bigg( \frac{ξ^{ℓ}}{(ℓ-1)!} + \sum_{j=0}^{ℓ-1} ℓ\frac{ξ^j}{j!}\bigg)\\
        &= ℓ - ℓe^{-ξ}\sum_{j=0}^{ℓ} \frac{ξ^j}{j!} = ℓQ(ξ,ℓ+1)\\
        ⇒ kℓ &= \frac{ξQ(ξ,ℓ)}{Q(ξ,ℓ+1)}.
    \end{align*}}
    The last equation characterises the threshold $c_{k,ℓ}^*$ for $ℓ$-orientability of random $k$-uniform hypergraphs, see for instance \cite{FKP:The_Multiple:2011}. Thus $c = c_{k,ℓ}^*$ follows.
    • We now make the dependence of $ϕ_c(u,v)$ on $c$ explicit. For monotonicity reasons we have $ϕ_c(u,v) > ϕ_{c'}(u,v)$ whenever $c < c'$ and $v ≠ 0$.
    Since $ϕ_{c_{k,ℓ}^*}$ is positive except for its two roots at $(0,0)$ and $(u₂,v₂)$, for $c < c_{k,ℓ}^*$ the potential $ϕ_c$ is positive except at $(0,0)$.
    • Since $ϕ_{c_{k,ℓ}^*}$ has a non-trivial root, $ϕ_c$ attains negative values for monotonicity reasons. By \itRef{it:extremum-is-sol}, the potential attains its (negative) minimum at a non-trivial solution to \cref{eq:fixed-point}, and by \itRef{it:smallest-sol-pos} it attains a positive value at the smallest non-trivial solution. Thus, the claim follows.
\end{enumerate}
\vspace{-\baselineskip}\phantom{a}\hfill 
\end{proof}
We are now ready to prove \cref{prop:convergence-of-q} by recruiting help from \cite{KRU:Wave-Like:2015}.

\begin{proof}[Proof of \cref{prop:convergence-of-q}]
    First note that we have $q₀ ≥ \P q₀$ by definition, which implies $\P^r q₀ ≥ \P^{r+1} q₀$ by monotonicity of $\P$ and induction on $r$. Thus, $\P^r q₀$ is pointwise bounded and decreasing and must converge to a limit $q^*$. As $\P$ is continuous (see \cref{lem:properties-hP}) we have $\P q^* = q^*$.
    \begin{enumerate}[(i)]
            • Let $𝟙 \colon ℝ → \{1\}$ be the $1$-function. First note that for any $x ∈ X$ we have, using properties from \cref{lem:properties-hP} and monotonicity of $h_f$ and $h_g$
            \begin{gather}
                (\P^r q₀)(x) ≤ (\hP^r 𝟙)(x) = (h_f ∘ h_g)^{r}(1)\\\notag
                \conv \max \{u ∈ [0,1] \mid h_f(h_g(u)) = u\}.\label{eq:trivial-upper-bound}
            \end{gather}
            So if the only solution of $h_f(h_g(u)) = u$ is $u = 0$, then we get $\P^r q₀(x) \conv 0$ from this alone. Otherwise, by \cref{lem:analysis-of-f-g} \itRef{it:two-sols}, there are one or two non-trivial solutions, the larger one we denote by $(u₂,v₂)$.
            
            We now apply \cite[Thm.\ 10]{KRU:Wave-Like:2015}%
            \footnote{Strictly speaking, the theorem requires functions $h_f$ and $h_g$ with $h_f(0) = h_g(0) = 0$ and $h_f(1) = h_g(1) = 1$. As the authors of \cite{KRU:Wave-Like:2015} point out themselves, this is purely to simplify notation. We can apply the theorem to our $h_f \colon [0,u₂] → [0,v₂]$ and $h_g \colon [0,v₂] → [0,u₂]$ with $h_f(0) = h_g(0) = 0$ and $h_f(u₂) = v₂, h_g(v₂) = u₂$ after rescaling the axes so $(u₂,v₂)$ becomes $(1,1)$. We will not do so explicitly.}. It requires $ϕ(u,v) > 0$ for $0 ≠ (u,v) ∈ [0,u₂] × [0,v₂]$, which we have shown in \cref{lem:analysis-of-f-g} \itRef{it:positive}. The theorem asserts pointwise convergence of $f^{(r)}$ to zero for any $f^{(0)}\colon ℝ → [0,u₂]$ in the case of one-sided termination. Clearly this implies convergence to zero in the case of two-sided termination as well, i.e.~$\P^r f^{(0)} \conv 0$. Choosing $f^{(0)} = 𝟙 · u₂$ we get
            \begin{align*}
                \lim_{r → ∞} (\P^r q₀) &= \lim_{r → ∞} \P^r \lim_{s → ∞} \P^s q₀\\
                &\refRel{eq:trivial-upper-bound}{≤} \lim_{r → ∞} \P^r f^{(0)} = 0.
            \end{align*}
            • Using \cref{lem:analysis-of-f-g} \itRef{it:non-positive} and \itRef{it:two-sols}, we know there are exactly three solutions $(0,0) < (u₁,v₁) < (u₂,v₂)$ to \cref{eq:fixed-point} and the signs of their potentials are zero, positive and negative, respectively. This is sufficient to apply \cite[Thm.\ 14]{KRU:Wave-Like:2015}\footnote{See previous footnote.}. The theorem asserts the existence of
            a solution $q^* \colon X → [0,u₂]$ of $\P q^* = q^*$ with $q^*(\frac{z+1}{2}) = u₂ - ε$ for any $ε > 0$, assuming $z = z(ε)$ is large enough.
            
            By monotonicity of $\P$ we have $\lim_{r → ∞} \P^r q₀ ≥ \lim_{r → ∞} \P^r q^* = q^*$.
    \end{enumerate}
    \vspace{-\baselineskip}\phantom{a}\hfill 
\end{proof}

%% file: interpolating-wave-to-peeling.tex

\section{Peelability of \texorpdfstring{$\bm{F_n}$}{Fn} below \texorpdfstring{$\bm{c_{k,ℓ}^*}$}{ckl*}}
\label{sec:wrapping-up}

We now connect the behaviour of system (\ref{eq:coupledSystem}) to the survival probabilities $\q{R}(x)$ we were originally interested in. For $c < c_{k,ℓ}^*$ and any $z ∈ ℕ$,  they can be made smaller than any $δ > 0$ in $R = R(δ,k,ℓ,z,c)$ rounds.

\begin{lemma}
    \label{lem:almost-peelable}
    If $c < c_{k,ℓ}^*$ then $∀z ∈ ℝ^+,δ > 0\colon ∃R,N ∈ ℕ\colon ∀n ≥ N, x ∈ X\colon \q{R}(x) < δ$.
\end{lemma}

\begin{proof}
    Let $z ∈ ℝ^+$ and $δ > 0$ be arbitrary constants. At first, \cref{prop:convergence-of-q} \ldl{(i)} implies only pointwise convergence $\P^r \q{0}(x) \conv 0$ for all $x ∈ X$. However, since $X$ is compact, since $\P^r \q{0}$ is continuous for $r > 0$ and since the all-zero limit is obviously continuous, basic calculus\footnote{Sometimes referred to as Dini's Theorem after Ulisse Dini (1848 – 1918).} implies uniform convergence, i.e.~there is a constant $R$ such that $\P^R \q{0}(x) ≤ δ/2$ for all $x ∈ X$. Therefore for $x ∈ X$:
    \begin{align*}
        \q{R}(x) &\stackrel{\text{Cor \ref{cor:approxOfq}}}= (\P^R \q{0}) (x) + o(1) ≤ δ/2 + o(1) ≤ δ.
    \end{align*}
    In the last step, we have simply choosen $N ∈ ℕ$ large enough.
\end{proof}
\cref{lem:localConvergenceOfH} only allows us to track $\q{R}$ via $\P^R q₀$ for a \emph{constant} number of rounds $R$. Therefore, we need to accompany \cref{lem:almost-peelable} with the following combinatorial argument that shows that if all but a $δ$-fraction of the vertices are peeled, then with high probability (whp; meaning probability $1-o(1)$) the rest is peeled as well. Arguments such as these are standard, many similar ones can be found for instance in \cite{FKP:The_Multiple:2011,FP:Sharp:2012,Luczak:A-simple-solution,L:A_New_Approach:2012,Luczak:Size-and-connectivity-of-the-k-core:1991,MPW:DoubleHashing:2018,Molloy05:Cores-in-random-hypergraphs}.

\begin{lemma}
    \label{lem:no-tiny-core}
    Let $c ∈ [0,ℓ]$. There exists $δ = δ(k,ℓ,z) > 0$ such that, whp, any subhypergraph of $F_n = F(n,k,c,z)$ induced by at most $δn$ vertices has minimum degree at most~$ℓ$.
\end{lemma}
\begin{proof}
    In the course of the proof, we will implicitly encounter positive upper bounds on $δ$ in terms of $k$, $ℓ$ and $z$. Any $δ > 0$ small enough to respect these bounds is suitable. We consider the bad events $W_{s,t}$ that some small set $V' ⊆ [n]₀$ of size $s$ induces $t$ hyperedges for $1 ≤ s ≤ δn$, $\frac{(ℓ+1)s}{k} ≤ t ≤ |E|$. If \emph{none} of these events occur, then all such $V'$ induce less than $(ℓ+1)|V'|/k$ hyperedges and therefore induce hypergraphs with average degree less than $ℓ+1$, so a vertex of degree at most $ℓ$ exists in each of them.
    
    We will show $\Pr[\bigcup_s \bigcup_t W_{s,t}] = \O(1/n)$ using a first moment argument. It is convenient to assume that hyperedges are $k$-tuples, possibly with repetition.
    First note that $F_n$ then contains three copies of the same hyperedge with probability at most $\binom{m}{3}(\frac{z+1}{n})^{-2k} = \O(n^{-2k+3}) = \O(n^{-1})$, so we restrict our attention to $F_n$ with at most two copies of the same hyperedge.
    Given $s$ and $t$ there are $\binom{n}{s}$ ways to choose $V'$. Since there are $s^k$ ways to form $k$-tuples from vertices of $V'$ and each hyperedge occurs at most twice, there are at most $\binom{2s^k}{t}$ multisets of hyperedges that $V'$ could induce. The probability that any given $k$-tuple occurs as a hyperedge is either zero if the $k$ vertices are too far apart or at most $1 - (1- (\frac{(z+1)}{n})^k)^{czn} ≤ \frac{(z+1)^kℓ}{n^{k-1}}$. Similarly, it occurs as a duplicate hyperedge with probability at most $(\frac{(z+1)^kℓ}{n^{k-1}})²$. Since the presence of hyperedges is negatively correlated, we may obtain an upper bound on the probability of the event that a set of hyperedges are all simultaneously present by taking the product of the events for the presence of the individual hyperedges. Thus, using constants $C, C', C'' ∈ ℝ^+$ (that may depend on $k,ℓ$ and $z$) where precise values do not matter, we get
    \begin{align*}
        &\phantom{=}\Pr[\bigcup_{s = 1}^{δn} \bigcup_{t = (ℓ+1)s/k}^{|E|} W_{s,t}]
        ≤ \sum_{s = 1}^{δn}\sum_{t = (ℓ+1)s/k}^{|E|}\Pr[W_{s,t}]\\
        &≤ \sum_{s = 1}^{δn}\sum_{t = (ℓ+1)s/k}^{|E|} \binom{n}{s} \binom{2s^k}{t} \bigg(\frac{(z+1)^kℓ}{n^{k-1}}\bigg)^t\\
        &≤ \sum_{s = 1}^{δn}\sum_{t = (ℓ+1)s/k}^{|E|} \bigg(\frac{\e n}{s}\bigg)^{s} \bigg(\frac{2\e (z+1)^kℓs^k}{tn^{k-1}}\bigg)^{t}\\
        &≤ \sum_{s = 1}^{δn}\sum_{t = (ℓ+1)s/k}^{|E|} \bigg(C\frac{n}{s}\bigg)^{s} \bigg(C'\frac{s^{k-1}}{n^{k-1}}\bigg)^{t}\\
        &≤ 2\sum_{s = 1}^{δn}\bigg(C\frac{n}{s}\bigg)^{s} \bigg(C'\frac{s^{k-1}}{n^{k-1}}\bigg)^{\lceil(ℓ+1)s/k\rceil}\\
        &= 2\sum_{s = 1}^{δn}\big(C''\tfrac{s}{n}\big)^{\lceil((k-1)(ℓ+1)-k)\frac{s}{k}\rceil}
        ≤ 2\sum_{s = 1}^{δn}\big(C''\tfrac{s}{n}\big)^{\lceil\frac{s}{k}\rceil}.
    \end{align*}
    To get rid of the summation over $t$, we assumed $(s/n)^{k-1} ≤ δ^{k-1} ≤ \frac{1}{2C'}$, in the last step we used $k ≥ 2$, $ℓ ≥ 1$ and $(k,ℓ) ≠ (2,1)$. Elementary arguments show that in the resulting bound, the contribution of summands for $s ∈ \{1,…,2k\}$ is of order $\O(\frac{1}{n})$, the contribution of the summands with $s ∈ \{2k+1,…,O(\log n)\}$ is of order $\O(\frac{\log n}{n²})$ (using $\frac{s}{n} ≤ \smash{\frac{\log n}{n}}$) and the contribution of the remaining terms with $s ≥ 3\log₂ n$ is of order $\O(2^{-\log₂ n}) = \O(\frac{1}{n})$  (using $C'' \frac{s}{n} ≤ C''δ ≤ \frac{1}{2}$).
    This gives $\Pr[\bigcup_{s,t} W_{s,t}] = \O(n^{-1})$, which proves the claim.
\end{proof}

We are now ready to prove the first half of \cref{thm:main}.

\begin{proof}[Proof of \cref{thm:main} (i)]
        Let $c < c_{k,ℓ}^*$ and $z ∈ ℝ^+$. We need to show that $F_n$ is $ℓ$-peelable whp.

        First, let $δ = δ(k,ℓ,z)$ be the constant from Lemma \ref{lem:no-tiny-core} and $R = R(δ/2)$ as well as $N$ the corresponding constants from \cref{lem:almost-peelable}.
        
        Assuming $n ≥ N$ we have $\q{R}(x) ≤ δ/2$ for all $x ∈ X$, meaning any vertex $v$ from $F_n$ is \emph{not} deleted within $R$ rounds of $\peel_{v,R}(F_n)$ with probability at most $δ/2$. Since $\peel(F_n)$ deletes in $R$ rounds at least the vertices that any $\peel_{v,R}(F_n)$ for $v ∈ V$ deletes in $R$ rounds, the expected number of vertices not deleted by $\peel(F_n)$ within $R$ rounds is at most $δn/2$.
        
        Now standard arguments using Azuma's inequality (see, e.g.~\cite[Thm.\ 13.7]{MU:Probability:2017}) suffice to conclude that whp at most $δn$ vertices are not deleted by $\peel(F_n)$ within $R$ rounds.
        
        By Lemma \ref{lem:no-tiny-core}, whp, neither the remaining $δn$ vertices, nor any subset induces a hypergraph of minimum degree at least $ℓ+1$. Therefore $\peel(F_n)$ deletes all vertices whp.
\end{proof}

%% file: peeling-needs-orientability.tex

\section{Non-Orientability of \headMath{F_n}{Fn} above \headMath{c_{k,ℓ}^*}{ckl*}}
\label{sec:upperbound}

To show that $F_n$ is not \emph{$ℓ$-peelable} whp for $c > c_{k,ℓ}^*$ we argue that $F_n$ is even not \emph{$ℓ$-orientable} whp.\footnote{Alternatively, one could try to base a proof on \cref{prop:convergence-of-q} (ii), possibly by going through similar motions as \cite[Lemma 4]{Molloy05:Cores-in-random-hypergraphs}. If successful, this might give a detailed characterisation of the \emph{$(ℓ+1)$-core} of $F_n$ – the largest subhypergraph of $F_n$ with minimum degree $ℓ+1$. Presumably, the $(ℓ+1)$-core contains roughly a $q^*(x)$-fraction of the vertices with position roughly at $x ∈ X$. We leave this aside. Our approach has the upside of establishing a connection between orientability thresholds and peelability thresholds.}
Our proof relies on local weak convergence theory, a subject we danced around in \cref{sec:peeling-operator}. There are three ingredients.

\subparagraph{Ingredient 1: Identical weak limits.}
For a finite graph $G$, let $G(∘)$ be the random rooted graph obtained by designating a root at random. For a rooted (possibly infinite) graph $T$, let $T(r)$ be the $r$-neighbourhood of the root.
\begin{Definition}[Random Weak Limit \cite{L:A_New_Approach:2012}\footnote{The name random weak limit comes from \cite{L:A_New_Approach:2012}. The notion is also known as Benjamini-Schramm limit \cite{BS:DistributionalLimits:2011}. Aldous and Steele \cite{AS:Objective_Method:2004} call it the \emph{standard construction}.}]
    Let $(G_n)_{n∈ℕ}$ be a sequence of (fixed) graphs and $T$ a random (possibly infinite) rooted graph. We say that $(G_n)_{n∈ℕ}$ has \emph{random weak limit} $T$ if $G_n(∘)(r)$ converges in distribution to $T(r)$ as $n → ∞$, for all $r ∈ ℕ$.
\end{Definition}
\def\Tv{T_{\mathrm{vert}}}
\def\Te{T_{\mathrm{edge}}}
For example, for $c ∈ ℝ^+$, $k ∈ ℕ$ and $n ∈ ℕ$, let $H_n = H_{n,cn}^k$ be the fully random $k$-uniform hypergraph. Let $G^H_n$ be the incidence graph of $H_n$. In particular, $G^H_n$ is bipartite with $cn$ vertices of degree $k$ that correspond to hyperedges in $H_n$ and $n$ vertices (of varying degrees) that correspond to vertices in $H_n$. Moreover, consider the random (possibly infinite) tree $\Tv$ generated as follows. The root vertex is on level zero. A vertex $v$ at an even level is given a random number $X_v \sim \Po(c)$ of children on the next level. A vertex at an odd level is given $k-1$ children on the next level. Let further $\Te$ be the random tree with a root connected to the roots of $k$ independently sampled copies of $\Tv$. Lastly, let $T$ be the random tree obtained by taking a copy of $\Tv$ with probability $\frac{1}{1+c}$ and a copy of $\Te$ with probability $\frac{c}{1+c}$.

The following claim is standard\footnote{I cannot find a crystal clear reference for this, but \cite{Leconte:Cuckoo:2013,L:A_New_Approach:2012} \emph{consider} it to be standard with reference to \cite{K:Poisson:2006}. Since the language differs significantly, I consider \cite{Leconte:Cuckoo:2013} itself to be a better reference, since an arguably more complicated case is treated in detail.}.
\begin{fact}
    \label{fact:fully-random-rwl}
    Almost surely, the sequence $(G^H_n)_{n∈ℕ}$ has random weak limit $T$.\footnote{
        It is easy to get confused here because we implicitly “cast” the sequence $(G^H_n)_{n∈ℕ}$ of random variables on graphs into a sequence of graphs. To reiterate: Having a certain random weak limit is a property of a \emph{sequence of graphs} (not of distributions). The claim is that when sampling a sequence of graphs by independently sampling each element $G^H_n$ of the sequence as explained above, then the resulting sequence of graphs will have the property almost surely, i.e.~with probability $1$.}
\end{fact}
\def\TF{\widetilde{F}_n}
\def\GF{\widetilde{G}^{F}_n}
Now, let also $z ∈ ℝ^+$ and let $F_n = F(n,k,c,z)$ be the random hypergraph from \cref{def:hypergraphs}.
We define $\TF$ to be a “borderless” version of $F_n$ where the vertices $i$ and $i+\frac{nz}{z+1}$ for all $i ∈ [\frac{n}{z+1}]₀$ are merged, “glueing” the right-most $\frac{n}{z+1}$ vertices of $F_n$ on top of the left-most $\frac{n}{z+1}$ vertices of $F_n$. Moreover, let $\GF$ be the incidence graph of $\TF$.

Techniques from \cite{Leconte:Cuckoo:2013} suffice to prove that we get the same random weak limit.
\begin{fact}
    \label{lem:limits-coincide}
    Almost surely, the sequence $(\GF)_{n∈ℕ}$ has random weak limit $T$.\footnote{Note that the limit $T$ does not depend on $z$.}
\end{fact}
\subparagraph{Ingredient 2: Lelarge's Theorem \cite{L:A_New_Approach:2012}.}
To clarify its role in our proof, we restate a remarkable theorem due to Lelarge \cite{L:A_New_Approach:2012} in weaker form.
\def\T{T^*}
\begin{theorem}[{Lelarge \cite[Theorem 4.1]{L:A_New_Approach:2012}}]
    \label{thm:lelarge}
    Let $(G_n = (A_n, B_n, E_n))_{n∈ℕ}$ be a sequence of bipartite graphs with $|E_n| = \O(|A_n|)$. Let further $M(G_n)$ be the maximum size of a set $E' ⊆ E_n$ with $\deg_{E'}(a) ≤ 1$ for $a ∈ A_n$ and $\deg_{E'}(b) ≤ ℓ$ for $b ∈ B_n$. If $(G_n)_{n∈ℕ}$ has random weak limit $\T$ and $\T$ satisfies certain natural properties\footnote{The limit $\T$ must be a \emph{bipartite unimodular Galton-Watson tree}, see \cite{L:A_New_Approach:2012} for an explanation. It is clear that $\T = T$ has the required properties.}, then $\smash{\lim\limits_{n→∞}} \frac{M(G_n)}{|A_n|}$ exists and depends only on $\T$.
\end{theorem}
A graph $G$ in the theorem should be interpreted as the incidence graph of a hypergraph $H$ with vertex set $B$ and hyperedge set $A$. Then $M(G)$ is the size of a largest set $A' ⊆ A$ such that the subhypergraph $(B,A')$ of $H$ is $ℓ$-orientable. In other words, $M(G)$ is the size of the largest \emph{partial $ℓ$-orientation} of $H$.
\subparagraph{Ingredient 3: Orientability-Gap above the threshold.}
Assume $c = c_{k,ℓ}^* + ε$ for $ε > 0$. By definition, it is not the case that $H_n$ is $ℓ$-orientable whp. More strongly however, it is known \cite{FKP:The_Multiple:2011,L:A_New_Approach:2012} that there exists a constant $δ = δ(ε) > 0$ such that the largest partial $ℓ$-orientation of $H_n$ has size $(1-δ)cn + o(n)$ whp. In the terms of \cref{thm:lelarge}, this means $\smash{\lim\limits_{n→∞}} {M(G^H_n)}/{|cn|} = 1-δ$ almost surely. We now put all three ingredients together.
\begin{proof}[Proof of \cref{thm:main} (ii)]
    Let $c = c_{k,ℓ}^* + ε$ and $δ = δ(ε)$ as above. We pick $z ≥ z^* ≔ \frac{2ℓ}{δc}$.
    
    Since $(G^H_n)_{n∈ℕ}$ and $(\GF)_{n∈ℕ}$ almost surely share the random weak limit $T$ by \cref{fact:fully-random-rwl,lem:limits-coincide}, we conclude from \cref{thm:lelarge} that the orientability gap carries over from $H_n$ to $\TF$, i.e.~${\lim\limits_{n→∞}} {M(\GF)}/m = 1-δ$ almost surely, where $m = \frac{czn}{z+1}$ is the number of hyperedges in $F_n$ and $\TF$.
    
    
    In particular, the size of the largest partial $ℓ$-orientation of $\TF$ is
    $(1-δ)m+o(n)$ whp. Switching from $\TF$ back to $F_n$ splits $\frac{n}{z+1}$ vertices and can increase the size of a largest partial $ℓ$-orientation by at most $\frac{ℓn}{z+1}$ to $(1-δ+\frac{ℓ}{c(z+1)})m+o(n) ≤ (1-\frac{δ}{2})m+o(n)$ whp. Thus $F_n$ is not $ℓ$-orientable whp.
\end{proof}

%% file: experiments.tex

\section{Experiments}
\label{sec:experiments}

To compare our new construction to other families of peelable hypergraphs, we used them to construct $1$-bit retrieval data structures as explained in \cref{sec:HBDS}. The following peeling-based variations have been implemented.

\begin{description}
        •[BPZ \cite{BPZ:Practical:2013}] $H$ is a fully random $3$-uniform hypergraph with a hyperedge density below the $1$-peelability threshold $\cp{3}{1} ≈ 0.818$. Construction via peeling and \eval-operations are very fast, but there is a sizeable overhead of $23\%$.
        •[LMSS \cite{LMSS:Efficient_Erasure:2001}] The hyperedges are distributed such that $H$ is the $1$-peelable hypergraph from \cite{LMSS:Efficient_Erasure:2001} already mentioned in \cref{sec:comparison}.
        To our knowledge, these hypergraphs have not been considered in the context of retrieval.
        •[COUPLED (this work).] The hyperedges are distributed such that $H = F(n,k,c,z)$. Recall that the hyperedge density $\hat{c} = \frac{m}{n}$ is slightly smaller than the chosen parameter $c$.
\end{description}
\begin{table*}[t]
    \includegraphics[page=4]{img/retrieval-exp.pdf}
    \caption{Overheads and average running times per key of various practical retrieval data structures.}
    \label{tab:results}
    \vspace{1em}
    \hrule
\end{table*}

\noindent Constructions of retrieval data structures not relying on peeling exist \cite{Vigna:Fast-Scalable-Construction-of-Functions:2016,DW:Retrieval-log-extra-bits:2019,DW:One-Block-per-Row:2019,Weaver:XORSAT-Filter,P:An_Optimal:2009} and we have implemented three of them. 
To avoid high construction times for solving non-peelable linear systems they use the standard trick of splitting the input into chunks of small expected size $C$ using an additional hash function.
\textsc{gov} \cite{Vigna:Fast-Scalable-Construction-of-Functions:2016} is similar to \textsc{bpz} in that it uses fully random $3$- or $4$-uniform hypergraphs, but \textsc{gov} allows for densities up to $\co{3}{1}$ or $\co{4}{1}$.
While being more complicated, at least “\textsc{2-block}” from \cite{DW:Retrieval-log-extra-bits:2019} undercuts the achievable overhead of the peeling-based approaches and at least “\textsc{1-block}” from \cite{DW:One-Block-per-Row:2019} offers better construction time due to being particularly cache-efficient. Results are reported in \cref{tab:results}. The used C++ code has been made available \cite{Walzer:RetrievalImplementation:2020} and a detailed discussion is found in the author's thesis \cite[Chap. 12]{W:Thesis}.

Experiments were run on a Microsoft Surface Pro 6 with an Intel Core i5-8250U Processor with a maximum single-core frequency of 3.40GHz. In all cases, the data set $S$ contains $m = 10^7$ random 64 bit integers.\footnote{When using the first $m = 10^7$ URLs from the \texttt{eu-2015-host} dataset gathered by \cite{BMSV:Crawls:2014} with ${≈}$80 bytes per key, we get similar results, except that running times are uniformly increased due to the higher cost for evaluating hash functions on these large keys.} The function $f \colon S → \{0,1\}$ is the parity of the integer\footnote{Note that the number and type of operation performed by \construct and \eval is completely independent of $f$. Thus $f$ does not affect our measurements.}. Hash functions are based on \texttt{xxhash} \cite{Collet:xxhash:2020} (we thereby depart from the full randomness assumption implicit in our theoretical analysis).
Query times are averages obtained by performing \eval for all elements of the data set once. The reported numbers are averages over 10 executions. Reported overheads are $1-N/m$ where $N ≥ n$ is the total number of bits used by the data structure, including varying amounts of metadata.

The main takeaway is that our new construction can beat the main other peeling-based approach BPZ in terms of overhead while offering similar running times. It seems likely that these improvements carry over to other data structures relying on peeling as well (see \cref{sec:HBDS}), for instance the space requirement of a construction of minimum perfect hash functions also suggested in \cite{BPZ:Practical:2013} drops from $≈ 2.44$ bits per key to $≈2.18$ bits per key (see \cite[Theorem A3]{W:Thesis}). With \textsc{lmss} it is possible to achieve smaller overheads at the expense of higher running times. Note also that the largest hyperedge size $D+4$ of \textsc{lmss} is exponential in the average hyperedge size. Therefore, the worst-case running time of \eval is much larger than the reported average running time of \eval.

\paragraph{On Choosing $\bm{z}$.} Our theoretical considerations treat $z$ as a constant and offer no guidance in how $z = z(n)$ should be chosen in practice. But consider the following heuristic argument, suggesting that $z = Θ(n^{1/3})$ may maximise the achievable hyperedge density $\hat{c} = \frac mn = c \frac{z}{z+1}$.

For finite $z$ and $n$, two things keep us from achieving $\hat{c} = \co{k}{1}$. \emph{Firstly}, vertices with positions $x ∈ [0,1] ∪ [z,z+1]$ at the borders have on average half the expected degree compared to other vertices. This leads to an overhead of $ε₁ = c - \hat{c} ≈ 1/z$. \emph{Secondly}, we need to choose $c = \co{k}{1} - ε₂$ slightly smaller than $\smash{\co{k}{1}}$. To see why, imagine the peeling process working linearly through the coupling dimension $X$. What happens around $x ∈ X$ should mostly depend on the expected $Θ(n/z)$ hyperedges with positions close to $x$. The standard deviation of their number is $Θ(\sqrt{n/z})$. If the local density is correspondingly increased by $Ω(\sqrt{n/z}/(n/z)) = Ω(\sqrt{z/n})$, this should not lead to the local density exceeding $\smash{\co{k}{1}}$ (otherwise we might get stuck), which calls for $ε₂ = \co{k}{1} - c ≥ Ω(\sqrt{z/n})$. Balancing $ε₁$ and $ε₂$ yields our recommendation of $z = Θ(n^{1/3})$.

Supporting experiments are found in \cref{fig:choosing-z}. There we estimate, for $n ∈ \{10⁴,10⁵,10⁶\}$ and various $z$ the value $\hat{c}(n,z)$ for which $F(n,3,\hat{c}(n,z)\frac{z+1}{z},z)$ is $1$-peelable with probability exactly $1/2$. To do so, we construct $500$ hypergraphs with distribution $F(n,3,1,z)$. Let the hyperedge set be $E = \{e₁,…,e_m\}$. We then determine the values
\[ m^* = \max_{m' ∈ [m]}\{ H' = ([n], \{e₁,…,e_{m'}\}) \text{ is $1$-peelable} \}.\]
This can be done by a customised peeling process on $H$ that, whenever no vertex of degree $1$ exists, deletes the hyperedge with highest index. Then $m^*+1$ is the highest index of a hyperedge that was deleted using the special rule. The median of the values $\frac{m^*}{n}$ over the $500$ runs is our approximation of $\hat{c}(n,z)$. The $z$ values maximising $\hat{c}(10⁴,z)$, $\hat{c}(10⁵,z)$ and $\hat{c}(10⁶,z)$ differ by roughly the factor $10^{1/3} ≈ 2.15$ as predicted by the argument above.
The choice of $z = 120$ for $m = 10⁷$ used in \cref{tab:results} is extrapolated from these observations.

\begin{figure}
    \includegraphics[page=1]{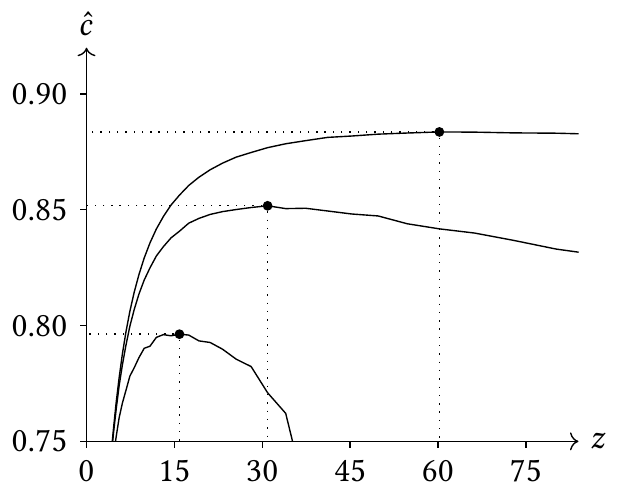}
    \caption[On choosing $z$ in finite spatially coupled hypergraphs.]{Approximate hyperedge densities $\hat{c}$ such that $F(n,3,\hat{c}\frac{z+1}{z},z)$ is $1$-peelable with probability $\frac 12$, for different choices of $z$. The three plot lines correspond, from bottom to top, to $n = 10⁴$, $n = 10⁵$, $n = 10⁶$.}
    \label{fig:choosing-z}
\end{figure}

Note that the results reveal that our construction does not downscale particularly well. For $n = 10^6$ the achieved densities are still more than $3\%$ below the optimum of $c_{k,ℓ}^*$ achieved for $n → ∞$.\footnote{One trick to mitigate this problem to some degree was suggested by Thomas Mueller Graf (personal communication) and is discussed in \cite[Chap. 12.2.5]{W:Thesis}.}

%% file: conclusion.tex

\section{Conclusion}
\label{sec:conclusion}

We have constructed families of $k$-uniform random hypergraphs with i.i.d.\ random hyperedges and an \emph{$ℓ$-peelability} threshold that approaches (for $z → ∞$) the \emph{$ℓ$-orientability} threshold $c_{k,ℓ}^*$ of fully random $k$-uniform hypergraphs.

We conjecture that this is best possible, i.e.\ no family of $k$-uniform random hypergraphs with i.i.d.\ random hyperedges has an $ℓ$-peelability threshold exceeding $c_{k,ℓ}^*$. In fact, even achieving $ℓ$-orientability beyond $c_{k,ℓ}^*$ seems unlikely.

We demonstrated the usefulness of our construction for hashing-based data structure using the example of retrieval data structures. The applicability of peelable hypergraphs is much wider, however, and whether using our construction yields significant improvements needs to be explored case by case. For instance, the stronger locality of the hyperedges might turn out to be advantageous in some settings while the higher number of rounds required by the (parallel) peeling process might be a problem in others.

We exploited the phenomenon of “threshold saturation via spatial coupling” that was discovered in coding theory and our proof borrows the powerful methods that were developed in that area.
We are very pleased to so effortlessly obtain improvements in hashing-based data structures and are curious to see whether this connection might be fruitful in other ways as well.